\documentclass[letterpaper, 10 pt, final]{IEEEtran}

\IEEEoverridecommandlockouts

\usepackage{amssymb,amsfonts,amsthm}
\usepackage{eucal,bm}
\usepackage{gensymb,subfigure}
\usepackage{booktabs}
\usepackage{times,color,url}
\usepackage{comment}

\usepackage{cite}

\usepackage{array}

\usepackage{graphicx}
\usepackage{epstopdf}

\usepackage[cmex10]{amsmath}
\usepackage{enumitem}

\newcommand{\f}[2]{\frac{#1}{#2}}

\newcommand{\h}[0]{{\mathcal R \mathcal H}_\infty}
\renewcommand{\d}[0]{\mathbb{D}}
\renewcommand{\c}[0]{\partial \mathbb{D}}
\newcommand{\dc}[0]{\overline{\mathbb{D}}}
\newcommand{\p}[0]{\mathcal{P}}

\newtheorem{theorem}{Theorem}
\newtheorem{lemma}{Lemma}
\newtheorem{corollary}{Corollary}

\allowdisplaybreaks

\title{Approximation by Simple Poles -- Part II: System Level Synthesis Beyond
  Finite Impulse Response}

\author{Michael W. Fisher, Gabriela Hug, and Florian D\"{o}rfler
  \thanks{This paper is based upon work supported by the King
  Abdullah University of Science and Technology (KAUST) Office of
  Sponsored Research (award No. OSR-2019-CoE-NEOM-4178.11) and by the
  European Union’s Horizon 2020 research and innovation program (grant
  agreement No. 883985).

  F. D\"{o}rfler and G. Hug are with
  ETH Z\"{u}rich, 8092 Z\"{u}rich, Switzerland.
  
  M. W. Fisher is with University of Waterloo, Waterloo, Ontario, Canada.

  Email: \{mfisher, dorfler\}@ethz.ch; hug@eeh.ee.ethz.ch.
}
}

\begin{document}

\maketitle

\begin{abstract}
  In Part I, a novel Galerkin-type method for finite dimensional approximations
  of transfer functions in Hardy space was developed based on approximation by
  simple poles.
  In Part II, this approximation is applied to system level
  synthesis, a recent approach based on a clever
  reparameterization, 
  to develop a new technique for optimal control design.
  To solve system level synthesis problems, prior work
  relies on finite impulse response
  approximations that lead to
  deadbeat control, and that can experience infeasibility and increased
  suboptimality, especially in systems with large separation of time scales.  
  The new design method
  does not result in deadbeat
  control, is convex and tractable, always feasible,
  can incorporate prior knowledge, and works well for systems
  with large separation of time scales.
  Suboptimality bounds with convergence rate depending on the geometry of the
  pole selection are provided.
  An example demonstrates
  superior performance of the method.
\end{abstract}

\section{Introduction}\label{sec:intro}

In Part I \cite{Fi22a} the approximation by a finite collection of transfer
functions with simple poles
was studied as a Galerkin-type method for
approximating transfer functions in Hardy space.
This paper applies this simple pole approximation (SPA) to optimal design of
linear feedback controllers.
A powerful approach for solving optimal control problems involves
convex reparameterization of stabilizing controllers, examples of which include
the Youla parameterization \cite{Yo76},
input-ouput parameterization (IOP) \cite{Fu19},
and system level synthesis (SLS) \cite{Wa18,Wa19}.
For this work, we focus on the closed-loop system responses for state
feedback controllers, and so restrict
our attention to SLS rather than Youla (which does not directly parameterize
using the closed-loop responses) or IOP (which focuses on output feedback).

The SLS reparameterization for mixed $\mathcal{H}_2/\mathcal{H}_\infty$
synthesis results in a convex but infinite dimensional optimization problem.
In order to solve it, prior work \cite{An19} has approximated that the
closed-loop responses are finite impulse responses (FIR) in order to arrive
at a tractable finite dimensional optimization problem.
However, this results in deadbeat control (DBC), which often experiences poorly
damped oscillations between discrete sampling times that can even
persist in steady state, as well as lack of robustness to model uncertainty and
parameter variations because of the high
control gains required to reach the origin in finite time \cite{Ur87}.
We denote SLS with the FIR approximation by DBC for the remainder
of the paper.




With DBC, the number of poles in the closed-loop transfer functions is
equal to the length of the FIR, potentially resulting in large numbers
of poles that can lead to high
computational complexity for the control design, lack of robustness in
the resulting controller, and implementation challenges in practice
\cite[Chapter 19]{Zh95}.
This is especially problematic when the optimal solution
has a long settling time, such as in systems with large separation of time
scales, where short sampling times are needed to capture the fast dynamics,
which are also coupled with much slower dynamics.
This leads to closed-loop impulse responses settling only after a large number
of time steps.
In addition, FIR closed-loop responses 
have all poles at the origin,
which results in infeasibility in case of stable but uncontrollable poles
in the plant. To resolve this, DBC introduces a slack variable enabling
constraint violation, 
which
leads to additional suboptimality \cite{An19}.
Furthermore, in this case DBC leads to a quasi-convex problem, requiring
an iterative approach such as golden section search to solve rather
a single convex optimization \cite{An19}.

This work combines SLS with SPA \cite{Fi22a} to develop a new control method
which addresses these limitations.
This approach is not FIR, so it does not suffer from the drawbacks of deadbeat
control.
Moreover, the number of poles is independent of the settling time of the
optimal closed-loop responses, and therefore SPA even works well for systems
with large separation of timescales.
It results in a convex and tractable optimization for the
design, avoiding the need for iterative methods,
requires only a small number of poles,
guarantees
feasibility 
for stabilizable systems without introducing slack variables, and
additional suboptimality resulting from these can be avoided.
Finally, if prior information is known about the optimal solution, such as the
locations of some of the optimal poles (e.g., for model matching 
\cite{Sh09},
model reference control \cite{Ab09}, design based on the internal model
principle \cite{Fr76}, expensive control \cite[Theorem~3.12(b)]{Kw72}, etc.),
then these
can be incorporated directly into the design for improved performance.

A suboptimality certificate is provided which shows the convergence
rate of SPA to the ground-truth optimal solution based on the geometry of
the pole selection.  Unlike a similar certificate for DBC, this does not
require a long enough time horizon for the optimal impulse response to
decay to be valid, and its convergence rate does not depend on this decay rate.
This certificate is then specialized to a particular pole selection
based on an Archimedes spiral as in \cite[Theorem~4]{Fi22a}.
An example shows superior performance of
SPA over DBC, and is fully reproducible with all code publicly
available \cite{git_sls}.


The paper is organized as follows.
Section~\ref{sec:back} provides preliminaries and problem setup,
Section~\ref{sec:review} reviews SLS,
Section~\ref{sec:res} provides the SPA method and suboptimality certificates,
Section~\ref{sec:ex} shows an illustrative example, 
Section~\ref{sec:proof} gives the proofs, 
and Section~\ref{sec:con} offers concluding remarks.



\section{Preliminaries}\label{sec:back}

We use the same notation as in Part I \cite{Fi22a},
and refer the reader to the preliminaries and main results sections there for
further details.  Recall also Assumptions A1-A5 from Part I.

\subsection{Problem Setup}

Consider the following LTI system in discrete time
\begin{align}
  \begin{split}
  x(k+1) &= Ax(k) + Bu(k) + \hat{B}w(k) \\
  y(k) &= Cx(k) + Du(k)
  \end{split}
  \label{eq:sys}
\end{align}
where $x(k) \in \mathbb{R}^n$, $u(k) \in \mathbb{R}^p$,
$w(k) \in \mathbb{R}^q$, and $y(k) \in \mathbb{R}^m$ are the state,
controller input,
disturbance input, and performance output vectors at time step
$k$, respectively.
Let $\sigma$ be the plant poles (i.e., the eigenvalues of $A$).
It will be useful to introduce the following related system
\begin{align*}
  \begin{split}
  x(k+1) &= Ax(k) + Bu(k) + v(k) \\
  y(k) &= Cx(k) + Du(k)
  \end{split}
\end{align*}
where $v(k) \in \mathbb{R}^n$ and the other signals are defined analogously to
\eqref{eq:sys}.
For signals $u(z)$ and $y(z)$ in the $z$-domain,
  let $T_{u \to y}(z)$ denote the transfer
function from $u(z)$ to $y(z)$.
Consider a linear (possibly dynamic) state feedback control law of the form
$u(z) = K(z)x(z)$ where $K \in \h$, and
let $T_{\text{desired}}(z)$ be some desired closed-loop
transfer function for model reference or model matching control
(note that we can set $T_{\text{desired}}(z) = 0$ if desired).
The goal is to choose a controller $K(z)$ that is a solution to the
mixed $\mathcal{H}_2$/$\mathcal{H}_\infty$ \cite{Zh95,Du00} optimal control
problem given by
  \begin{align}
    \begin{split}
    & \min_{K(z)}
      \left|\left|
      T_{w \to y}(z) - T_{\text{desired}}(z) \right|\right|_{\mathcal{H}_2} \\
      & \quad \quad
      + \lambda
      \left|\left|
      T_{w \to y}(z) - T_{\text{desired}}(z) \right|\right|_{\mathcal{H}_\infty} \\
    & \;\text{s.t.} \quad  T_{v \to x}(z), T_{v \to u}(z) \in \f{1}{z} \h,
    \label{eq:goal}
    \end{split}
  \end{align}
  where $\lambda \in [0,\infty]$ is constant.
  As $T_{w \to y}(z)$ is nonconvex in $K(z)$, \eqref{eq:goal} is known to be a
  challenging problem.
  We make the following feasibility assumption:

  \vspace{5pt}

  { \it
  \indent (A6) A solution to \eqref{eq:goal} exists,
  i.e., $(A,B)$ is stabilizable, \indent and the optimal closed-loop transfer
  functions
  are rational \indent (hence they have finitely many poles).}

  \vspace{5pt}  
  
  While one can construct pathological examples where this assumption does not
  hold (e.g., a controllable SISO system with $y = x$ and
  $T_{\text{desired}}(z) = e^{\f{1}{z}}$),
  in the standard mixed $\mathcal{H}_2$/$\mathcal{H}_\infty$ setting
  Assumption A6 is satisfied automatically \cite{Wa94}.


  By Assumption A6 there exists an optimal solution $(T_{v \to x}^*,T_{v \to u}^*)$
  to \eqref{eq:goal}.
  As $T_{v \to x}^*, T_{v \to u}^* \in \f{1}{z}\h$, we can write their partial
  fraction decomposition as
  \begin{align}
    \begin{split}
      T_{v \to u}^*(z) &= \sum\nolimits_{q \in \mathcal{Q}} \sum\nolimits_{j=1}^{m_q^*}
      H_{(q,j)}^*
      \f{1}{(z-q)^j} \\
      T_{v \to x}^*(z) &= \sum\nolimits_{q \in \hat{\mathcal{Q}}}
      \sum\nolimits_{j=1}^{\hat{m}_q} G_{(q,j)}^*
      \f{1}{(z-q)^j}
      \label{eq:phi_opt}
    \end{split}
  \end{align}
  where $\mathcal{Q}$ and $\hat{\mathcal{Q}}$ are finite sets of stable poles
  closed under complex conjugation, $H_{(q,j)}^*$ and $G_{(q,j)}^*$ are
  coefficient matrices,
  and $m_q^*$ and $\hat{m}_q$ are the
  multiplicities of the pole $q$ in $T_{v \to u}^*$ and $T_{v \to x}^*$,
  respectively.
  It will be shown (in the proof of Lemma~\ref{lem:bound})
    that the following relationship
    between the poles $\mathcal{Q}$ and $\hat{\mathcal{Q}}$ holds:
    $\hat{\mathcal{Q}} \subset \mathcal{Q} \cup \sigma$.
    Thus, each pole of $T_{v \to x}^*$ must be a pole of at least one of
    $T_{v \to u}^*$ and the plant.

\section{Review of System Level Synthesis} \label{sec:review}

\subsection{System Level Parameterization}    
  
A recent approach was proposed to solve problem \eqref{eq:goal}
for the special case,
  where
  $y = \left[\begin{smallmatrix} (Qx)^\intercal & (Ru)^\intercal \end{smallmatrix}\right]^\intercal$
  for constant matrices $Q$ and $R$,
  $T_{\text{desired}}(z) = 0$, and $\hat{B} = I$.
  This approach is known as system level synthesis (SLS) \cite{An19}, and
  the key idea is to reparameterize the control design in terms of the
  closed-loop transfer functions $\Phi_x(z) = T_{v \to x}(z)$ and
  $\Phi_u(z) = T_{v \to u}(z)$.
  This transforms \eqref{eq:goal} into an infinite dimensional convex
  optimization
  problem at the price of an additional affine constraint 
  (further details are given in \cite{An19}).
  Noting that $T_{w \to y}(z) = C T_{v \to x}(z)\hat{B} + DT_{v \to u}(z)\hat{B}$,
  this results in
  \begin{align}
  \begin{split}
    \min_{\Phi_x(z),\Phi_u(z)} &
    \left|\left|
    C\Phi_x(z)\hat{B} + D\Phi_u(z)\hat{B} - T_{\text{desired}}(z)
    \right|\right|_{\mathcal{H}_2} \\
    & + \lambda \left|\left|
    C\Phi_x(z)\hat{B} + D\Phi_u(z)\hat{B} - T_{\text{desired}}(z)   
    \right|\right|_{\mathcal{H}_\infty} \\
    \text{s.t.} \quad & (zI-A)\Phi_x(z) - B\Phi_u(z) = I \\
    & \Phi_x(z), \Phi_u(z) \in \f{1}{z} \h.
    \label{eq:spa_inf}
  \end{split}
  \end{align}
  This is a strict generalization of \cite[Eq.~4.35]{An19} to
  \eqref{eq:goal}.
  After solving \eqref{eq:spa_inf}, a 
  controller that yields the optimal closed-loop responses 
  can be
  recovered via
  $K(z) = \Phi_u(z)\Phi_x^{-1}(z)$, and realizations of $K(z)$ exist which do
  not require
  transfer function inversion (for more details see \cite{An19}).
  

  

  \subsection{Finite Impulse Response Approximation} \label{sec:fir}
  
  To obtain a tractable optimization problem, the FIR approximation is made
  for the closed-loop transfer functions $\Phi_x$ and $\Phi_u$ \cite{An19},
  i.e.,
  $\Phi_x(z) = \sum\nolimits_{i=1}^T G_i \f{1}{z^i}$
  and $\Phi_u(z) = \sum\nolimits_{i=1}^T H_i \f{1}{z^i}$,
  where $G_i$ and $H_i$ are coefficient matrices.
  For any positive integer $T$, we call this DBC with FIR length $T$.
  If the plant is uncontrollable, then it has stable poles which cannot be
  removed by feedback,
  so it is infeasible to achieve FIR closed-loop transfer functions.
  To maintain feasibility, 
  DBC introduces a slack variable $V$
  that allows the affine constraints to be violated.
  However, the objective becomes non-convex as a result,
  so
  DBC uses a quasi-convex upper bound of the objective \cite{An19}.
  The resulting control design
  is quasi-convex and finite dimensional, 
  and can be solved
  using methods such as golden section search.
  The true (i.e., realized) closed-loop responses are then given by
      $T_{v \to x}(z) = \Phi_x(z)\left(I + \f{V}{z^T}\right)^{-1}$ and 
      $T_{v \to u}(z) = \Phi_u(z)\left(I + \f{V}{z^T}\right)^{-1}$ \cite{An19}.

  \subsection{System Level Synthesis Certificates}\label{sec:sls_cert}
  
  Let $J^*$ be the ground-truth optimal cost of \eqref{eq:goal}, and let
  $J(T)$ be the optimal cost of DBC with an FIR of length $T$.
  Let $(\Phi_x^*,\Phi_u^*)$ be an optimal solution to \eqref{eq:goal}.
  Then there exist constants $C_*, \rho_* > 0$ such that
  $||\mathcal{I}(\Phi_x^*)(k)||_2 \leq C_*\rho_*^k$ for all $k \geq 0$.
  Then $T$ sufficiently large such that $C_*\rho_*^T < 1$ is sufficient for DBC
  to be feasible
  and 
  to satisfy the following
  suboptimality bound \cite[Theorem~4.7]{An19} for some $c > 0$, which is shown
  as a relative error bound for ease of comparison:
  \begin{align}
    \f{J(T)-J^*}{J^*} \leq  \f{C_*\rho_*^T}{1- C_*\rho_*^T}
    \left(1 + \f{\lambda c}{1-\rho_*^T}\right).
    \label{eq:sls_bound}
  \end{align}

  When $\rho_*$ is small (i.e. the optimal closed-loop response is
  slow), such as for systems with large separation of
  time scales, $C_* \rho^T < 1$ may require large $T$, the convergence rate
  of $C_* \rho^T$ in \eqref{eq:sls_bound} is slow, and the term
  $\f{1}{1-C_*\rho_*^T}$ in \eqref{eq:sls_bound} (which arises from the slack
  variable $V$) will further slow convergence.

\section{Main Results}\label{sec:res}

\subsection{Simple Pole Approximation (SPA) Control Design}
\label{sec:spa}


To introduce our new method, we begin by reformulating
\eqref{eq:goal} using the SLS reparameterization, which results in the
following convex but infinite dimensional optimization problem
which is a strict generalization of the formulation in \cite{An19}:

Recall that $\sigma$ are the poles of the plant, where each $q \in \sigma$ has
multiplicity $m_q$, and $\p$ represents a selection
of poles within the unit disk \cite{Fi22a}.
To obtain a tractable optimization problem, we approximate $\Phi_x$ and
$\Phi_u$ using $\mathcal{P}$ and $\sigma$ by
\begin{align}
  \begin{split}
  \Phi_u(z) &= \sum\nolimits_{p \in \mathcal{P}} H_p \f{1}{z-p} \\  
  \Phi_x(z) &= \sum\nolimits_{p \in \mathcal{P}} G_p \f{1}{z-p}
  + \sum\nolimits_{q \in \sigma} \sum\nolimits_{i=1}^{m_q+1} G_{(q,i)} \f{1}{(z-q)^i}
  \label{eq:spa_approx}
  \end{split}
\end{align}
where $H_p$, $G_p$, and $G_{(q,i)}$ are coefficient matrices.
We refer to this as the simple pole approximation (SPA) since all of the poles
other than the poles of the plant in $\Phi_x$ are simple.
As we will see, the poles in $\sigma$ are included in $\Phi_x$
with multiplicities potentially greater than one in order to ensure feasibility
in case the plant is stabilizable but not controllable.
However, 
it is not necessary (though
possible if desired) to include
poles with multiplicity greater than one in $\Phi_u$ as well, which is
why there is an asymmetry in the approximations of $\Phi_u$ and $\Phi_x$ in
\eqref{eq:spa_approx}.
Note that the coefficients for the plant poles range from $1$ to $m_q+1$ in
$\Phi_x$, since if $\Phi_u$ has a pole at the same location as the plant, and
because $\Phi_u$ has only simple poles, it is
possible to increase the multiplicity of this pole by one.

Although it is possible to select any poles $\mathcal{P} \subset \d$
for the SPA method, we provide several recommendations that often lead to
improved performance.  First, we suggest to include the poles of the plant
$\sigma$ in $\p$ to allow the design to cancel out any
controllable modes of the plant for which it is advantageous to do so.
In addition, for any poles of the optimal solution which are
known a priori (see Section~\ref{sec:intro}), including these in $\p$ can
lead to a dramatic improvement in performance.
For the remaining poles, the Archimedes spiral is a natural choice as it
provides an approximately even
pole selection over $\d$ and converges at the rate $(|\p|+2)^{-1/2}$
\cite{Fi22a}.


For any $q \in \sigma$, let $\tilde{m}_q = 1$ if $q \in \p$ and
$\tilde{m}_q = 0$ otherwise.
Then the SPA of \eqref{eq:spa_approx} applied to \eqref{eq:spa_inf}
results in the following optimal control design problem,
consisting of the objective
  \begin{align}
    \begin{split}
      \min_{H_p,G_p,G_{(q,i)}} &
      \left|\left|
      \mathcal{I}(C\Phi_x\hat{B})
      + \mathcal{I}(D\Phi_u\hat{B}) - \mathcal{I}(T_{\text{desired}})
      \right|\right|_F \\
    & + \lambda \left|\left|
      \mathcal{C}(C\Phi_x\hat{B})
      + \mathcal{C}(D\Phi_u\hat{B}) - \mathcal{C}(T_{\text{desired}})
      \right|\right|_2,
    \end{split} \label{eq:spa_obj}
  \end{align}
  subject to the following SLS constraints (whose form given below is
  derived in the proof of Lemma~\ref{lem:bound}):
  \begin{align}
    \begin{split}
    G_{(q,2)} + (qI-A)G_{(q,1)} 
    -BH_q &= 0, \quad 
    \forall~ q \in \sigma \cap \p \\
    (pI-A)G_p - BH_p &= 0, \quad \forall~ p \in \p - \sigma \\
    G_{(q,i+1)} + (qI-A)G_{(q,i)} &= 0, 
    \quad
    \forall~
    q \in \sigma, \\
    & \quad \quad \quad \;  i \in \{1+\tilde{m}_q,...,m_q\} \\
    (qI-A)G_{(q,m_q+\tilde{m}_q)} &= 0, \quad \forall~ q \in \sigma \\
    \sum\nolimits_{p \in \p - \sigma} G_p + \sum\nolimits_{q \in \sigma} G_{(q,1)} &= I
    \end{split} \label{eq:spa_sls}
  \end{align}
  and the impulse responses
  \begin{align}
    \begin{split}
    \mathcal{I}(\Phi_u)(k) &= \sum_{p \in \p} p^{k-1} H_p \\
    \mathcal{I}(\Phi_x)(k) &= \sum_{p \in \p-\sigma} p^{k-1} G_p
    + \sum_{q \in \sigma} \sum_{i=1}^{m_q+1} p^{k-i} {k-1 \choose i-1} G_{(q,i)}.
    \end{split}
    \label{eq:spa_imp}    
  \end{align}
  It is straighforward to see that in \eqref{eq:spa_sls}-\eqref{eq:spa_imp}
  the SLS constraints
  and the impulse responses are affine and linear,
  respectively, in the coefficients $H_p$, $G_p$, and $G_{(q,i)}$.
  As the impulse responses $\mathcal{I}$ and convolution operators
  $\mathcal{C}$ appearing in the objective \eqref{eq:spa_obj} are linear in
  the impulse responses
  of $\Phi_x$ and $\Phi_u$, this implies that the terms inside the norms
  $||\cdot||_F$ and $||\cdot||_2$ are affine in the coefficients
  $H_p$, $G_p$, and $G_{(q,i)}$.
  Therefore, since $||\cdot||_F$ and $||\cdot||_2$ are convex, the SPA
  control design \eqref{eq:spa_obj}-\eqref{eq:spa_imp} is convex.

  As representations of $\mathcal{I}$ and $\mathcal{C}$ would require matrices
  of infinite size, in order to evaluate the norms $||\cdot||_F$ and
  $||\cdot||_2$ in the objective \eqref{eq:spa_obj} in practice, we introduce
  a finite $T > 0$ and
  replace all instances of $\mathcal{I}$ and $\mathcal{C}$ in \eqref{eq:spa_obj}
  by $\mathcal{I}_T$ and $\mathcal{C}_T$, respectively
  (see \cite{Fi22a} for the notation).
  Then these norms become the standard Frobenius and spectral matrix norms,
  so \eqref{eq:spa_obj}-\eqref{eq:spa_imp} can be formulated as a tractable
  semidefinite program (SDP), and as a quadratic program (QP) in the special
  case of $\mathcal{H}_2$ design (i.e., $\lambda = 0$).
  As the dimension of \eqref{eq:spa_obj}-\eqref{eq:spa_imp} is independent of
  the time horizon $T$, in practice one can take $T$ sufficiently large such
  that the Frobenius and spectral norms in the objective approximate
  arbitrarily well the true $\mathcal{H}_2$ and $\mathcal{H}_\infty$ norms,
  respectively.

  Since the uncontrollable poles of the plant are included in $\Phi_x(z)$,
  feasibility is ensured whenever $(A,B)$ is stabilizable.
  As feasibility is guaranteed in this case and
  \eqref{eq:spa_obj}-\eqref{eq:spa_imp} is convex,
  unlike with DBC there is no
  need to introduce a slack
  variable or use iterative unimodal optimization methods
  for SPA.
  Instead, SPA can be solved with a single convex optimization (a SDP or QP),
  and then the closed-loop responses are given by 
  $T_{v \to x}(z) = \Phi_x(z)$ and $T_{v \to u}(z) = \Phi_u(z)$,
  which do not require inverting transfer functions as in
  Section~\ref{sec:review} for DBC.
  Furthermore, note that the poles in $\mathcal{P}$ can be chosen to lie
  anywhere within the open unit disk, so this method does not result in FIR
  closed-loop transfer functions and, hence, avoids deadbeat control.
  
  
\subsection{Suboptimality Bounds}\label{sec:subopt}

Recall that $d(z,\p)$ is the distance from $z$ to $\p$, and that
$D(\p) = \max_{z \in \d} d(z,\p)$ measures the worst case geometric
approximation error
between approximating poles $\p$ and optimal poles $\mathcal{Q}$.
In addition, $r \in (0,1)$ is such that $\p \subset \overline{B}_r$,
and $\delta$ is a measure of the minimum distance between each approximating
pole in $\p$ and $\sigma$ (see \cite{Fi22a} for further details).
Also, recall Assumptions A1-A5 from Part I \cite{Fi22a}.
Our main theoretical result shows that the relative error of the SPA method
decays at least linearly with $D(\mathcal{P})$.

\begin{theorem}[General Suboptimality Bound]\label{thm:gen}
Let $J^*$ denote the optimal cost of \eqref{eq:spa_inf}, and let
$J(\p)$ denote the optimal cost of
\eqref{eq:spa_obj}-\eqref{eq:spa_imp} for any choice of $\p$.
Suppose Assumption A6 is met, and $\p$ satisfies Assumptions A1-A5.
Then there exists a constant
$\hat{K} = \hat{K}(\mathcal{Q},G_{(q,j)}^*,H_{(q,j)}^*,r,\delta) > 0$
such that
\begin{align}
  \f{J(\mathcal{P}) - J^*}{J^*} \leq \hat{K} D(\mathcal{P}).
  \label{eq:spa_bound}
\end{align}
\end{theorem}

While the DBC suboptimality bound in \eqref{eq:sls_bound} only holds for
$T$ sufficiently large such that $||\Phi_x^*(T)||_2 \leq C_*\rho^T < 1$,
the SPA bound in \eqref{eq:spa_bound} does not have this requirement.
Furthermore, the DBC bound includes a term $\f{1}{1-C_*\rho_*^T}$ resulting
from the slack variable, whereas the SPA bound has no such term because it
does not need a slack variable.
Finally, the convergence for the DBC bound depends on the rate of decay
of the optimal closed-loop impulse response, whereas the SPA bound convergence
depends on the distance between
$\mathcal{P}$ and the optimal closed-loop poles.
Therefore, SPA is preferable when the optimal impulse response takes long
to decay, such as in stabilizable systems with large separation of time scales.
In addition, if some optimal poles can be included in $\p$ due to prior
knowledge (see Section~\ref{sec:spa}), this will typically
have the effect of decreasing both $D(\p)$ and $\hat{K}$
in \eqref{eq:spa_bound}, significantly reducing the relative error of SPA.
In contrast, it is not clear how such prior knowledge could be included with
DBC to reduce the relative error in \eqref{eq:sls_bound}.



     


Corollary~\ref{cor:spiral} shows that, for the Archimedes spiral pole selection
in \cite{Fi22a},
the relative error of SPA converges to zero at the rate $(|\p|+2)^{-1/2}$
since $|\p_n| = 2n-2$ for each $n > 0$.

\begin{corollary}[Spiral Suboptimality Bound]\label{cor:spiral}
  Consider the setup of Theorem~\ref{thm:gen} with pole selection $\p_n$ given
  as in \cite[Theorem~4]{Fi22a} for each 
  integer $n > 0$.
  Then there exists a constant
  $\hat{K} = \hat{K}(\mathcal{Q},G_{(q,j)}^*,H_{(q,j)}^*) > 0$ and $N > 0$
  such that $n \geq N$ implies
\begin{align}
  \f{J(\p_n) - J^*}{J^*} \leq \f{\hat{K}}{\sqrt{n}}.
\end{align}
\end{corollary}

Note that $N$ in Corollary~\ref{cor:spiral} only needs to be chosen to ensure
that $D(\p_N) < 1$, $|\p_N| \geq m_{\max}$, and that $\delta > 0$ for $\p_N$,
and so is typically satisfied in practice with small $N$
(see the remark following \cite[Theorem~3]{Fi22a}).





\section{Numerical Example}\label{sec:ex}



To compare DBC and SPA, we consider the example of using a
power converter to provide frequency and voltage control services to the power
grid, which arises
naturally as a result of interfacing
renewable generation 
to the grid \cite{Li20}. 
This example served as the motivation to develop the SPA
method, because of the inadequate performance of DBC resulting from the large
separation of time scales in power systems containing power converter
interfaced devices \cite{Su21}.
Let $w$ represent the frequency and voltage magnitude 
at the connection point, and let $y$ represent the power output of the
converter.
Then this can be formulated 
in the form of \eqref{eq:spa_inf} with matrices given by
\begin{align*}
  A &= \left[\begin{smallmatrix}
    0.988 & 0 & 0 & 0 & 0 \\
    0 & 0 & 0 & 0 & 0 \\
    1 & 0 & 0 & 0 & 0 \\
    0 & 0 & 0 & 0.995 & 0 \\
    0 & 0 & 0 & 0 & 0.9\end{smallmatrix}\right] \mkern-5mu, \;
  B = \left[\begin{smallmatrix}
    0 & 0 \\
    0 & 0 \\
    0 & 0 \\
    0.005 & 0 \\
    0 & 0.1\end{smallmatrix}\right] \mkern-5mu, \;
  D =
\left[\begin{smallmatrix}
    0 & 0 \\
    0 & 0 \\
    0.01 & 0 \\
    0 & 0.01\end{smallmatrix}\right] \\
  \hat{B} &= \left[\begin{smallmatrix}
    -0.0001 & 0 \\
    0 & 1 \\
    0.0066 & 0 \\
    0 & 0 \\
    0 & 0\end{smallmatrix}\right]\mkern-5mu, \quad
  C = \left[\begin{smallmatrix}
    0 & 0.829 & -0.428 & 1.02 & 0 \\
    0 & 0.428 & 0.829 & 0 & -1.02 \\
    0 & 0 & 0 & 0 & 0 \\
    0 & 0 & 0 & 0 & 0
    \end{smallmatrix}\right]
  \\
  &T_{\text{desired}}(z) = \left[\begin{smallmatrix}
    T_{\text{desired}}^{\hat{y}}(z) \\
    0
  \end{smallmatrix}\right]
  = \left[\begin{smallmatrix}
     \f{-5.3e-5}{z-0.999} & 0 \\
     0 & \f{-100}{z-0.999} \\
     0 & 0 \\
     0 & 0
  \end{smallmatrix}\right].
\end{align*}
For ease of comparison to the ground-truth optimal solution we choose
$\lambda = 0$,
and we note that an infinite impulse response method for SLS exists for
this special case \cite{Yu21}, but we emphasize that similar results to those
shown here hold for
$\lambda \neq 0$ (although the exact ground-truth optimal solution is difficult
to obtain).
With $\lambda = 0$, the objective is to minimize
\begin{align*}
  \left|\left| \left[\begin{smallmatrix}
    T_{w \to \hat{y}}(z) - T_{\text{desired}}^{\hat{y}}(z) \\
    0.01~T_{w \to u}(z)
  \end{smallmatrix}\right]
  \right|\right|_{\mathcal{H}_2}
\end{align*}
where $\hat{y} = Cx$.
Let $t \in \mathbb{R}$ and $k \in \mathbb{Z}$ denote continuous and
discrete time, respectively, with a sample time
of $h = 1~\text{ms}$ chosen to avoid aliasing from the fast converter dynamics.

To solve \eqref{eq:goal}, for DBC we use golden section search as suggested
in \cite[p. 380]{An19}, which involves
solving SDPs iteratively to find $\Phi_x$ and $\Phi_u$.
Then, $T_{w \to y}$ can be recovered from $\Phi_x$ and $\Phi_u$
by inverting $\left(I + \f{V}{z^T}\right)$ (see Section~\ref{sec:review})
and then applying a linear transformation.
For SPA, we let $\p$ consist of the poles of the plant
and $T_{\text{desired}}$, and select the remaining poles from
the Archimedes spiral as in \cite[Theorem~4]{Fi22a}.
Then, solving \eqref{eq:goal} using SPA only requires a single SDP,
and then $T_{w \to y}$ is given by a linear transformation of $\Phi_x$
and $\Phi_u$, so no transfer function inversion is necessary.
To solve the SDPs in each case, Matlab was used with YALMIP
and the solver MOSEK.  This control design implementation is available online
\cite{git_sls}.



The DBC and SPA control design approaches are run for varying numbers of poles.
For DBC, the problem is infeasible for 30 or less poles, 
converges in 16 (golden section) iterations for 31 poles, and converges in
7 iterations for 300 poles.
Recovering $T_{w \to y}$ from the DBC solution requires inverting a transfer
function, 
but for large numbers
of poles this leads to out of memory errors and numerical errors.
Therefore, the figures show only the result of using $\Phi_x$ and $\Phi_u$
in place of the true system responses $T_{v \to x}$ and $T_{v \to u}$,
so the true DBC results are actually worse than the DBC results shown in these
figures.
The SPA method is feasible for any number of poles, and requires only one SDP
for each number of poles. It is run for 7 and 15 poles, and
the true system responses are easily recovered.

The impulse responses for the solutions of DBC, SPA, and the optimal and desired
transfer functions are shown in Fig.~\ref{fig:h2_zoom}.
For DBC, the impulse response is close to the optimal impulse response only
for the first $31~\text{ms}$ or $300~\text{ms}$ for 31 and 300 poles,
respectively, after which the impulse response becomes zero (an undesirable but
inevitable feature of DBC).
However, the optimal impulse response takes several seconds to decay, so
overall the matching is very poor for DBC, with 300 poles only slightly better
than with 31 poles.
In contrast, for SPA the impulse response shows a large initial mismatch during
the first few milliseconds, but after this the matching is much closer, with
the 15 pole solution showing significantly better matching than the 7 pole case.
Note that the large initial impulse responses of the SPA method could be
reduced with additional convex constraints or by adding fast poles.
From the impulse responses it is clear that SPA is much closer to the optimal
solution, and with orders of magnitude fewer poles.


\begin{figure}[t]
  \centering
  \includegraphics[width=0.49\textwidth]{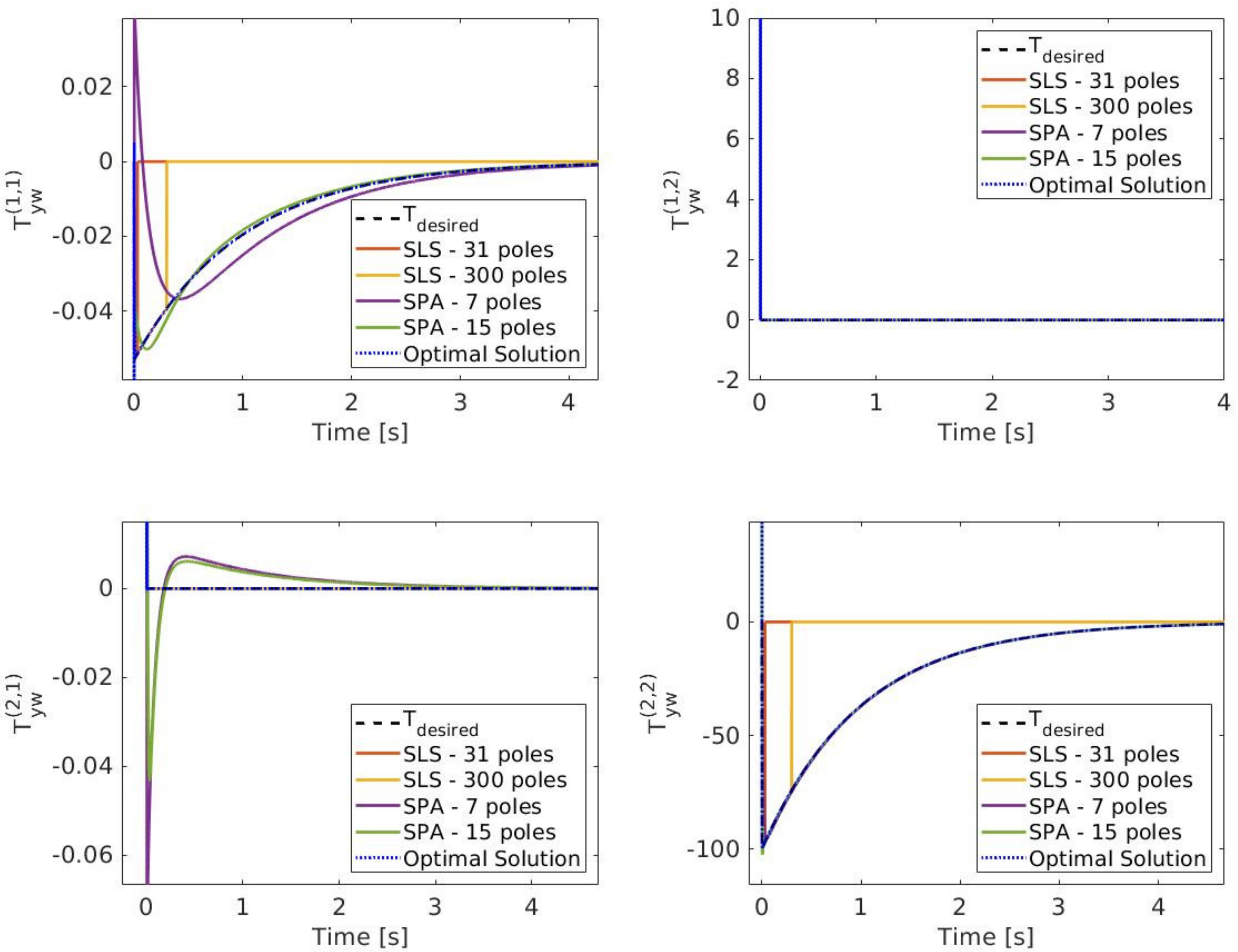}
  \caption{The impulse responses of control designs for
    SLS with the FIR approximation (DBC) and simple pole approximation (SPA)
    as a
    function of the number of poles in the controller.  The desired
    transfer function $T_{\text{desired}}$
    and the ground-truth optimal solution
    are also shown.}
  \label{fig:h2_zoom}
\end{figure}

The step responses for the solutions of DBC, SPA, and the optimal and desired
transfer functions are shown in Fig.~\ref{fig:h2_step}.
For DBC, the step responses deviate greatly from the optimal step response,
with the 300 pole solution closer than with 31 poles.
With SPA the step responses are close to the optimal step response, with
the 15 pole solution showing closer matching during the initial transient
than the 7 pole case.
However, even with SPA there is a small steady state error.
This could be removed either by imposing convex constraints on the
DC gain directly, or by changing the objective 
to the difference in the step responses (which is also convex). 
From the step responses it is clear that SPA results in much closer matching
with the optimal transfer function than DBC, and with far fewer poles.





\begin{figure}
  \centering
  \includegraphics[width=0.49\textwidth]{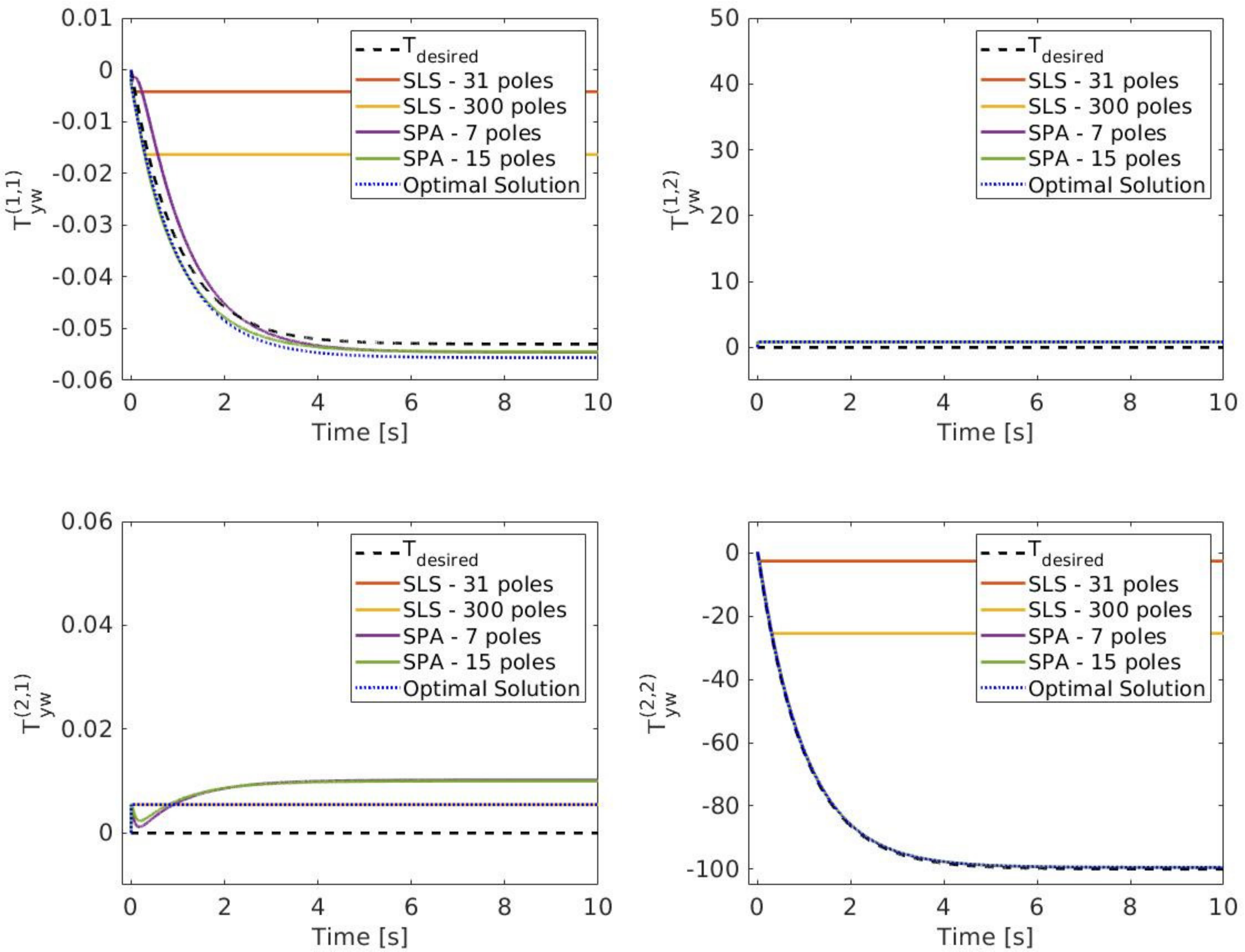}
  \caption{Step responses of control designs for SLS with the FIR approximation
    (DBC) and simple pole approximation (SPA) as a function of the number of
    poles in the controller.  The desired transfer function
      $T_{\text{desired}}$ and the ground-truth optimal solution are also shown.}
  \label{fig:h2_step}
\end{figure}

\section{Proofs}\label{sec:proof}

The key technical result required to prove Theorem~\ref{thm:gen}
is Lemma~\ref{lem:bound}, which extends the approximation error bounds of
\cite[Theorem~1]{Fi22a} to bound the error between a particular feasible
solution $(\Phi_u, \Phi_x)$ of \eqref{eq:spa_obj}-\eqref{eq:spa_imp} and the
optimal solution $(\Phi_u^*, \Phi_x^*)$ of \eqref{eq:spa_inf}.

\begin{lemma}\label{lem:bound}
  Under the conditions of Theorem~\ref{thm:gen}, let
  $(\Phi_x^*,\Phi_u^*)$ denote the optimal solution to \eqref{eq:spa_inf}.
  Then there exist $\Phi_x, \Phi_u \in \f{1}{z} \h$ which are a feasible
  solution to \eqref{eq:spa_obj}-\eqref{eq:spa_imp}, and constants
  $K^u_\infty, K^u_2, K^x_\infty, K^x_2 > 0$, such that
  \begin{align}
    ||\Phi_u-\Phi_u^*||_{\mathcal{H}_\infty} &\leq K^u_\infty D(\p), \quad
    ||\Phi_u-\Phi_u^*||_{\mathcal{H}_2} \leq K^u_2 D(\p) \label{eq:ubound} \\    
    ||\Phi_x-\Phi_x^*||_{\mathcal{H}_\infty} &\leq K^x_\infty D(\p), \quad
    ||\Phi_x-\Phi_x^*||_{\mathcal{H}_2} \leq K^x_2 D(\p). \label{eq:xbound}
  \end{align}
\end{lemma}

Before proving Lemma~\ref{lem:bound}, we will require several technical
results as given in the next few lemmas and corollaries.

\begin{lemma}\label{lem:abs}
Let $k$ be any integer, $m$ a positive integer, and $z \in \dc$.
Let $q, p_1,...,p_m \in \d$,
and let $\hat{d}(q) = \max_i |p_i-q|$.
If $z \in \c$, let $\delta = d(\c,\{p_i\}_{i=1}^m) > 0$ and
$\eta = d(q,\c) > 0$; if not, suppose
$d\left(z,\{p_i\}_{i=1}^m\right) \geq \delta > 0$
and $d(z,q) \geq \eta > 0$.
Then there exists $K > 0$ such that
\begin{align*}
  \left|\f{(z-q)^k}{\prod_{i=1}^m (z-p_i)} - (z-q)^{k-m} \right|
  &\leq K\hat{d}(q).
\end{align*}
\end{lemma}

\begin{proof}[Proof of Lemma~\ref{lem:abs}]
Let $\p = \{p_1,...,p_m\}$. We compute
\begin{align*}
  \left|\f{(z-q)^k}{\prod\limits_{i=1}^m (z-p_i)} - (z-q)^{k-m} \right| 
  = \f{\left|(z-q)^m - \prod\limits_{i=1}^m (z-p_i)\right|}
    {\left|(z-q)^{m-k}\prod\limits_{i=1}^m (z-p_i)\right|}.
\end{align*}
Noting that the proofs of \cite[Eqs. 10, 12]{Fi22a} are still valid for
$z \in \dc$ (i.e. $|z| \leq 1$), applying them here we have that
\begin{align*}
  \left|(z-q)^m - \prod_{i=1}^m (z-p_i)\right|
  \leq \left((|q|+2)^m - (|q|+1)^m\right) \hat{d}(q).
\end{align*}
Furthermore,
\begin{align*}
  \left|(z-q)^{m-k}\prod_{i=1}^m (z-p_i)\right|
  \geq d(z,q)^{m-k} d(z,\mathcal{P})^m
  \geq \eta^{m-k} \delta^m.  
\end{align*}
Combining these two inequalities implies that
\begin{align*}
  \left|\f{(z-q)^k}{\prod\limits_{i=1}^m (z-p_i)} - (z-q)^{k-m} \right|
  &\leq K \hat{d}(q) \\
  K = \f{\left((|q|+2)^m - (|q|+1)^m\right)}{\eta^{m-k} \delta^m}.
\end{align*}
\end{proof}

\begin{lemma}\label{lem:easy}
  Let $k$ be a nonnegative integer, $m$ a positive integer
  and $q, p_1, ..., p_m \in \d$.
  Choose constants $c_{p_i}$ as in the proof of \cite[Lemma~3]{Fi22a}.
  Then
  \begin{enumerate}[label=\alph*.]
  \item For $k < m$
    \begin{align*}
      \sum_{i=1}^m (p_i-q)^k c_{p_i} \f{1}{z-p_i} =
      \f{(z-q)^k}{\prod\limits_{i=1}^m (z-p_i)}.
    \end{align*}
  \item For $k \geq m$
    \begin{align*}
      \sum_{i=1}^m (p_i-q)^k c_{p_i} \f{1}{z-p_i}
      &= \f{(z-q)^k}{\prod\limits_{i=1}^m (z-p_i)}
      - \sum_{i=0}^{k-m} b_i (z-q)^i \\
      b_i &= \sum_{j=1}^m (p_j-q)^{k-1-i} c_{p_j}.
    \end{align*}
  \end{enumerate}
\end{lemma}

\begin{proof}[Proof of Lemma~\ref{lem:easy}]
First consider Case (a).
Write the partial fraction decomposition
\begin{align*}
  \f{(z-q)^k}{\prod_{i=1}^m (z-p_i)} = \sum\nolimits_{i=1}^m \kappa_i \f{1}{z-p_i}.
\end{align*}
Multiplying both sides by $\prod_{i=1}^m (z-p_i)$ and evaluating at $z = p_i$
implies that
\begin{align*}
  \kappa_i = \f{(p_i-q)^k}{\prod_{\substack{j = 1 \\ j \neq i}}^m (p_i-p_j)}
  = (p_i-q)^k c_{p_i}
\end{align*}
which completes the proof for Case (a).

Next consider Case (b).
Write the partial fraction decomposition
\begin{align}
  \f{(z-q)^k}{\prod_{i=1}^m (z-p_i)} = \sum\nolimits_{i=1}^m \kappa_i \f{1}{z-p_i}
  + \sum\nolimits_{i=0}^{k-m} b_i (z-q)^i.
  \label{eq:zqk}
\end{align}
Multiplying by $\prod_{i=1}^m (z-p_i)$ and evaluating at $z = p_i$ implies
\begin{align*}
  \kappa_i = \f{(p_i-q)^k}{\prod_{\substack{j=1 \\ j \neq i}}^m (p_i-p_j)}
  = (p_i-q)^k c_{p_i}.
\end{align*}
Differentiating \eqref{eq:zqk} $i$ times with respect to $z$ and evaluating at
$z = q$ implies that
\begin{align*}
  0 = -i! \sum\nolimits_{j=1}^m (p_j-q)^{k-1-i} c_{p_j} + i! b_i
\end{align*}
for $i \in \{0,...,k-m\}$, so
\begin{align*}
  b_i  = \sum\nolimits_{j=1}^m (p_j-q)^{k-1-i} c_{p_j}.
\end{align*}

\end{proof}

Let $m$ and $n$ be integers, and
define the rising factorial $m^{(n)} = \prod_{k=0}^{n-1} (m+k)$
and the falling factorial $m_n = \prod_{k=0}^{n-1} (m-k)$.
Note that for $m$ and $n$ nonnegative,
letting $m!$ denote the standard factorial, we have
$m^{(n)} = \f{(m+n-1)!}{(m-1)!}$ and $m_n = \f{m!}{(m-n)!}$.
It is straightforward to verify: \\
Fact 3.  $(-1)^{n} m^{(n)} = (-m)_n$ \\
Fact 4. $\sum_{j=0}^n {n \choose j} m_j (m')_{n-j} = (m+m')_n$.


\begin{lemma}\label{lem:hard}
  Let $k$ and $m$ be positive integers, $z \in \c$, 
  and $q, \lambda, p_1, ..., p_m \in \d$ with
  $d(\lambda,\{p_i\}_{i=1}^m) \geq \delta > 0$ and $d(\lambda,q) \geq \eta > 0$.
  Choose constants $c_{p_i}$ as in the proof of \cite[Lemma~3]{Fi22a}. Then
  \begin{enumerate}[label=\alph*.]
  \item  There exists $K > 0$ such that
\begin{align*}
  \sum_{i=1}^m c_{p_i} \f{(\lambda-p_i)^{-k}}{z-p_i}
  &= \f{(\lambda-z)^{-k}}{\prod\limits_{i=1}^m (z-p_i)}
  - \f{r(z)}{(\lambda-z)^k} \\
  r(z) &= \sum_{n=0}^{k-1} a_n (\lambda-z)^n
\end{align*}
and
\begin{align*}
  &\lim_{z \to \lambda} \f{d}{dz^l} \left(r(z) \prod\limits_{i=1}^m (z-p_i)\right) \\
  &= \begin{cases}
    -1, & l = 0 \\
    0, & l \in \{1,...,k-1\} \\
    \left(\f{(-1)^{l+1}m^{(l)}}{(\lambda-q)^{m+l}} + \epsilon\right)
    \prod\limits_{i=1}^m (\lambda-p_i),  & l = k
  \end{cases} \\
  &|\epsilon| \leq K \hat{d}(q).
\end{align*}
\item There exist $K_0',...,K_{k-1}' > 0$ such that
\begin{align*}
  &\sum_{i=1}^m c_{p_i} \f{(p_i-q)^m}{(\lambda-p_i)^{k}} \f{1}{z-p_i}
  = \f{(z-q)^m}{(\lambda-z)^{k}\prod\limits_{i=1}^m (z-p_i)} \\
  & \mkern+200mu \;
  - \sum_{n=0}^{k-1} \f{a_n}{(\lambda-z)^{k-n}} \\
  &|a_0-1| \leq K_0' \hat{d}(q), \;
  |a_n| \leq K_n' \hat{d}(q), \; n \in \{1,...,k-1\}.
\end{align*}
  \end{enumerate}
\end{lemma}

\begin{proof}[Proof of Lemma~\ref{lem:hard}]
For $l \in \{0,m\}$, write the partial fraction decomposition
\begin{align}
  \begin{split}
  \f{(z-q)^l}{(\lambda-z)^k \prod_{i=1}^m (z-p_i)}
  &= \sum\nolimits_{i=1}^m \kappa_i \f{1}{z-p_i} + \f{r(z)}{(\lambda-z)^k} \\
  r(z) &= \sum\nolimits_{n=0}^{k-1} a_n (\lambda-z)^n.
  \end{split}
  \label{eq:pfd}
\end{align}
Multiplying both sides by $(\lambda-z)^k \prod_{i=1}^m (z-p_i)$ yields
\begin{align}
  (z-q)^l &=
  (\lambda-z)^k\sum_{i=1}^m \kappa_i \prod_{\substack{j=1 \\ j\neq i}}^m (z-p_j)
  + r(z) \prod_{i=1}^m (z-p_i).
  \label{eq:equal}
\end{align}
Evaluating \eqref{eq:equal} at $z = p_i$ implies that
\begin{align*}
  \kappa_i =  c_{p_i} (p_i-q)^l (\lambda-p_i)^{-k}.
\end{align*}
For $n$ any nonnegative integer, define
\begin{align}
  b_n &= \left(\f{d}{dz^n} r(z) \right)(\lambda) = (-1)^n n! a_n \label{eq:an}\\
  d_n &= \left(\f{d}{dz^n} \prod\nolimits_{i=1}^m (z-p_i) \right)(\lambda)
  = \sum\limits_{\substack{v \in \mathbb{R}^n \\ v_i \in I_m \forall i
      \\ v_i \neq v_j \text{ for } i \neq j}}
  \prod\limits_{\substack{k \in I_m  \\ k \not\in v}} (\lambda-p_k) \\
  e_n &= \left(\f{d}{dz^n} (z-q)^l \right)(\lambda) 
  = l_n (\lambda-q)^{l-n}. \label{eq:en}
\end{align}
Note that
\begin{align}
  \lim_{z \to \lambda} \f{d}{dz^n} \left(r(z)\prod\nolimits_{i=1}^m(z-p_i)\right)
  = \sum\nolimits_{j=0}^n {n \choose j} d_j b_{n-j}
  \label{eq:zl}
\end{align}
for any nonnegative integer $n$.
Differentiating \eqref{eq:equal} $n$ times with respect to
$z$, and evaluating at $z = \lambda$ implies that
\begin{align}
  \begin{split}
  e_n &= \lim_{z \to \lambda} \f{d}{dz^n} \left(r(z)\prod\limits_{i=1}^m(z-p_i)\right) 
  = \sum_{j=0}^n {n \choose j} d_j b_{n-j}
  \end{split}
  \label{eq:dzn}
\end{align}
for $n \in \{0,...,k-1\}$.
Dividing by $d_0$ and solving for $b_n$ implies
\begin{align}
  b_n = \f{e_n}{d_0} - \sum\nolimits_{j=1}^n {n \choose j} \f{d_j}{d_0} b_{n-j}.
    \label{eq:bn}
\end{align}
Note that
\begin{align*}
  \f{d_n}{d_0} =
  \sum\limits_{\substack{v \in \mathbb{R}^n \\ v_i \in I_m \forall i
      \\ v_i \neq v_j \text{ for } i \neq j}}
  \prod\limits_{k \in v} \f{1}{\lambda-p_k}.
\end{align*}
Define
\begin{align*}
  \epsilon_n' &= \f{d_n}{d_0} - \f{m_n}{(\lambda-q)^n} 
  = \mkern-20mu
  \sum\limits_{\substack{v \in \mathbb{R}^n \\ v_i \in I_m \forall i
      \\ v_i \neq v_j \text{ for } i \neq j}}
  \left(\prod\limits_{k \in v} \f{1}{\lambda-p_k} - \f{1}{(\lambda-q)^n}\right)
\end{align*}
since the number of terms in the sum is $m_n$.
Thus, by Lemma~\ref{lem:abs} there exists $k_n' > 0$ such that
\begin{align}
  \begin{split}
  \f{d_n}{d_0} &= m_n \f{1}{(\lambda-q)^n} + \epsilon_n', \quad
  |\epsilon_n'| \leq k_n' \hat{d}(q).
  \end{split}
  \label{eq:dn}
\end{align}
Consider first Case (a): $l = 0$.  Then $e_0 = 1$ and $e_n = 0$ for
$n \in \{1,...,k-1\}$.
By \eqref{eq:dzn}, this implies the desired result for $n \in \{0,...,k-1\}$,
so it suffices to prove the desired result for $n = k$.
By \eqref{eq:bn}, $b_0 = \f{1}{d_0}$.
We claim that there exists $k_n > 0$ such that
\begin{align}
  \begin{split}
 -\sum\nolimits_{j=1}^n {n \choose j} \f{d_j}{d_0} b_{n-j}
 &=  \f{(-1)^n m^{(n)}}{(\lambda-q)^{m+n}} + \epsilon_n, \quad
  |\epsilon_n| \leq k_n \hat{d}(q)
  \end{split}
  \label{eq:bnd}
\end{align}
for all $n \in \{1,...,k\}$.
Note that by \eqref{eq:bn}, this implies that
\begin{align}
  \begin{split}
  b_n &= (-1)^n m^{(n)} \f{1}{(\lambda-q)^{m+n}} + \epsilon_n, \quad
  |\epsilon_n| \leq k_n \hat{d}(q)
  \end{split}
  \label{eq:help}
\end{align}
for $n \in \{1,...,k-1\}$.
We prove \eqref{eq:bnd} by strong induction.
For the base case, first note that
\begin{align}
  \begin{split}
  b_0 &= \f{1}{d_0} = \f{1}{(\lambda-q)^m}
  + \left(\f{1}{\prod_{i=1}^m (\lambda-p_i)} - \f{1}{(\lambda-q)^m} \right) \\
  &= \f{1}{(\lambda-q)^m} + \epsilon_0, \quad
  |\epsilon_0| \leq k_0 \hat{d}(q)
  \end{split}
  \label{eq:b0}
\end{align}
where such $k_0 > 0$ exists by Lemma~\ref{lem:abs}.
Then for $n=1$ we have
\begin{align*}
  - \f{d_1}{d_0} b_0 &= - \left(\f{m}{\lambda-q} + \epsilon_1'\right)
  \left(\f{1}{(\lambda-q)^m} + \epsilon_0\right) \\
  &= - \f{m}{(\lambda-q)^m} -  \f{m}{\lambda-q}\epsilon_0
  -\f{1}{(\lambda-q)^m} \epsilon_1' - \epsilon_0 \epsilon_1' \\
  &= - \f{m}{(\lambda-q)^m} + \epsilon_1, \quad
  |\epsilon_1| \leq k_1 \hat{d}(q) \\
  k_1 &= \f{mk_0}{|\lambda-q|} + \f{k_1'}{|\lambda-q|^m} + k_0k_1'.
\end{align*}
For the induction step, assume that \eqref{eq:bnd} holds for
all $j \in \{1,...,n-1\}$, which, together with \eqref{eq:b0},
implies that \eqref{eq:help} holds for all $j \in \{0,...,n-1\}$.
By \eqref{eq:dn} and \eqref{eq:help} we have
\begin{align*}
  &- \sum_{j=1}^n {n \choose j} \f{d_j}{d_0} b_{n-j}
    = - \sum_{j=1}^n {n \choose j}
      \left(m_j \f{1}{(\lambda-q)^j} + \epsilon_n'\right) \\
      &
      \mkern+80mu *
      \left((-1)^{n-j} m^{(n-j)} \f{1}{(\lambda-q)^{m+n-j}} + \epsilon_{n-j}\right)
      \\
      &= - \sum\nolimits_{j=1}^n {n \choose j} m_j (-1)^{n-j} m^{(n-j)}
      \f{1}{(\lambda-q)^{m+n}}
      + \epsilon_n \\
      &\epsilon_n = - \sum_{j=1}^n {n \choose j} 
      \left(\f{m_j\epsilon_{n-j}}{(\lambda-q)^j} 
      +  \f{(-1)^{n-j} m^{(n-j)}\epsilon_n'}{(\lambda-q)^{m+n-j}}
      + \epsilon_n'\epsilon_{n-j} \right)\\
      &|\epsilon_n| \leq k_n \hat{d}(q) \\
      &k_n = \sum\nolimits_{j=1}^n {n \choose j}
      \left(\f{m_jk_{n-j}}{|\lambda-q|^j}  
      + \f{m^{(n-j)}k_n'}{|\lambda-q|^{m+n-j}} 
      + k_{n-j}k_n'\right).
\end{align*}
So 
\begin{align*}
  &- \sum\nolimits_{j=1}^n {n \choose j} \f{d_j}{d_0} b_{n-j} \\
  &
  \overset{\substack{\text{above} \\ \text{identity}}}{=}  
  - \f{1}{(\lambda-q)^{m+n}}
  \sum\nolimits_{j=1}^n {n \choose j} m_j (-1)^{n-j} m^{(n-j)}
  + \epsilon_n \\
  &
  \overset{\text{Fact } 3}{=}
  - \f{1}{(\lambda-q)^{m+n}}
  \sum\nolimits_{j=1}^n {n \choose j} m_j (-m)_{n-j}
  + \epsilon_n  \\
  &
  \overset{\text{add } 0}{=}
  \f{1}{(\lambda-q)^{m+n}}
  \left((-m)_n - \sum_{j=0}^n {n \choose j} m_j (-m)_{n-j}  \right)
  + \epsilon_n \\
  &
  \overset{\text{Fact } 4}{=}
  \f{1}{(\lambda-q)^{m+n}} \left((-m)_n - (m-m)_{n} \right)
    + \epsilon_n \\
  &
  \overset{0_n = 0}{=}    
    \f{1}{(\lambda-q)^{m+n}} (-m)_n + \epsilon_n \\
  &
  \overset{\text{Fact } 3}{=}
  (-1)^n m^{(n)} \f{1}{(\lambda-q)^{m+n}} + \epsilon_n, \quad
  |\epsilon_n| \leq k_n \hat{d}(q).
\end{align*}
Thus, \eqref{eq:bnd} holds.
Note that $b_k = 0$ since $r(z)$ is a polynomial of order $k-1$.
Therefore, by \eqref{eq:zl} and \eqref{eq:bnd} we have that
\begin{align*}
  &\lim_{z \to \lambda} \f{d}{dz^k} \left(r(z)\prod\nolimits_{i=1}^m(z-p_i)\right)
    = \sum\nolimits_{j=0}^k {k \choose j} d_j b_{k-j} \\
    &= d_0b_k + \sum\nolimits_{j=1}^k {k \choose j} d_j b_{k-j}
    = \sum\nolimits_{j=1}^k {k \choose j} d_j b_{k-j} \\
    &= (-d_0)\left(-\sum\nolimits_{j=1}^k {k \choose j}
    \f{d_j}{d_0} b_{k-j}\right) \\
    &= (-d_0) \left((-1)^k m^{(k)} \f{1}{(\lambda-q)^{m+k}} + \epsilon_k \right),
    \quad |\epsilon_k| \leq k_k \hat{d}(q)
\end{align*}
which yields the result for Case (a).
Next consider Case (b): $l = m$.
Then by \eqref{eq:bn} and \eqref{eq:an}
\begin{align*}
  a_0 = b_0 = \f{e_0}{d_0} = \f{(\lambda-q)^m}
  {\prod\nolimits_{i=1}^m (\lambda-p_i)}
\end{align*}
so
\begin{align*}
  |a_0 - 1| = \left| \f{(\lambda-q)^m}
  {\prod\nolimits_{i=1}^m (\lambda-p_i)} - 1\right|
  \leq k_0 \hat{d}(q)
\end{align*}
where such $k_0 > 0$ exists by Lemma~\ref{lem:abs}.
We claim that there exist $k_n > 0$ such that
\begin{align}
  |b_n| \leq k_n \hat{d}(q)
  \label{eq:bnd2}
\end{align}
for $n \in \{1,...,k-1\}$.
We prove \eqref{eq:bnd2} by strong induction.
For the base case, note that by \eqref{eq:bn} and \eqref{eq:dn}
\begin{align*}
  &b_1 = \f{e_1}{d_0} - \f{d_1}{d_0} b_0 
  = \f{m (\lambda-q)^{m-1}}{d_0} -
  \left(m \f{1}{\lambda-q} + \epsilon_1' \right)
  \f{e_0}{d_0} \\
  &=  \f{m (\lambda-q)^{m-1}}{d_0} - m \f{1}{\lambda-q} \f{(\lambda-q)^m}{d_0}
  - \epsilon_1' \f{(\lambda-q)^m}{\prod\nolimits_{i=1}^m (\lambda-p_i)} \\
  &=  \f{m (\lambda-q)^{m-1}}{d_0} - \f{m(\lambda-q)^{m-1}}{d_0}
  - \epsilon_1' \f{(\lambda-q)^m}{\prod\nolimits_{i=1}^m (\lambda-p_i)} \\
  &= - \epsilon_1' \f{(\lambda-q)^m}{\prod\limits_{i=1}^m (\lambda-p_i)}, \quad
  |b_1| \leq k_1 \hat{d}(q), \quad
  k_1 = k_1' \f{|\lambda-q|^m}{\prod\limits_{i=1}^m |\lambda-p_i|}.
\end{align*}
For the induction step, assume \eqref{eq:bnd2} holds for all
$j \in \{1,...,n-1\}$.
By \eqref{eq:bn}, \eqref{eq:dn}, \eqref{eq:en}, and the induction hypothesis
we have
\begin{align*}
  b_n &
  \overset{\eqref{eq:bn}}{=}
  \f{e_n}{d_0} - \sum\nolimits_{j=1}^n {n \choose j} \f{d_j}{d_0} b_{n-j} \\
  &
  \overset{\substack{\text{regrouping} \\ \text{terms}}}{=}    
  \f{e_n}{d_0} - \f{d_n}{d_0} b_0
  - \sum\nolimits_{j=1}^{n-1} {n \choose j} \f{d_j}{d_0} b_{n-j} \\
  &
  \overset{\substack{\eqref{eq:en} \\ \eqref{eq:dn}}}{=}    
  \f{m_n (\lambda-q)^{m-n}}{d_0}
  - \left(m_n \f{1}{(\lambda-q)^n} + \epsilon_n'\right)
  \f{e_0}{d_0} \\
  &- \sum\nolimits_{j=1}^{n-1} {n \choose j} \left(m_j \f{1}{(\lambda-q)^j}
  + \epsilon_j' \right) b_{n-j} \\
  &
  \overset{\substack{\eqref{eq:en} \\ \text{regrouping}}}{=}    
  \f{m_n (\lambda-q)^{m-n}}{d_0}
  -  \f{m_n (\lambda-q)^{m-n}}{d_0} \\
  &-  \f{\epsilon_n'(\lambda-q)^m}{\prod\limits_{i=1}^m (\lambda-p_i)}
  - \sum\nolimits_{j=1}^{n-1} {n \choose j} \left(m_j (\lambda-q)^{-j}
  + \epsilon_j' \right) b_{n-j} \\
  &
  \overset{\text{cancel}}{=}      
  -  \f{\epsilon_n'(\lambda-q)^m}{\prod\limits_{i=1}^m (\lambda-p_i)}
  - \sum_{j=1}^{n-1} {n \choose j} \left(m_j (\lambda-q)^{-j} + \epsilon_j'
  \right) b_{n-j} \\
  |b_n| &\leq k_n \hat{d}(q) \\
  k_n &= k_n' \f{|\lambda-q|^m}{\prod\limits_{i=1}^m |\lambda-p_i|}
  +\sum_{j=1}^{n-1} {n \choose j} \left(m_j|\lambda-q|^{-j} + k_j' \right) k_{n-j}.
\end{align*}
Thus, \eqref{eq:bnd2} holds.
Combining \eqref{eq:bnd2} with \eqref{eq:pfd} and \eqref{eq:an} yields the
result for Case (b).
\end{proof}

For $z \in \mathbb{C}$, let $J(z)$ denote an elementary Jordan
  block with eigenvalue $z$ whose dimension can be inferred from context.

\begin{corollary}\label{cor:hard}
  Let $k$ be a nonnegative integer, $m$ a positive integer,  $z \in \c$, 
  and $q, \lambda, p_1, ..., p_m \in \d$ with
  $d(\lambda,\{p_i\}_{i=1}^m) \geq \delta > 0$ and $d(\lambda,q) \geq \eta > 0$.
  Choose constants $c_{p_i}$ as in the proof of \cite[Lemma~3]{Fi22a}.  Then
  \begin{enumerate}[label=\alph*.]
    \item There exists $K > 0$ such that
\begin{align*}
  \sum_{i=1}^m \f{c_{p_i}}{(\lambda-p_i)^{(k+1)}}
  &= {m-1 + k \choose k} (\lambda-q)^{-(m+k)} + \epsilon \\
  |\epsilon| \leq K\hat{d}(q).
\end{align*}
\item There exists $K > 0$ such that
\begin{align*}
  \left|\left| \sum_{i=1}^m c_{p_i} (p_i-q)^m J(\lambda-p_i)^{-1} \f{1}{z-p_i}
  \right|\right|_2 \leq K\hat{d}(q).
\end{align*}
  \end{enumerate}
\end{corollary}

\begin{proof}[Proof of Corollary~\ref{cor:hard}]
First we prove Case (a).
By Lemma~\ref{lem:hard}(a) we have
\begin{align*}
  \sum\nolimits_{i=1}^m c_{p_i} (\lambda-p_i)^{-(k+1)}
  &= \lim_{z \to \lambda} \sum\nolimits_{i=1}^m c_{p_i} (\lambda-p_i)^{-k}
  \f{1}{z-p_i} \\
  &= \lim_{z \to \lambda}
  \f{1- r(z) \prod\nolimits_{i=1}^m (z-p_i)}
    {(\lambda-z)^k \prod\nolimits_{i=1}^m (z-p_i)}.
\end{align*}
Furthermore, by Lemma~\ref{lem:hard}(a), the numerator and denominator have
the limits
\begin{align*}
  \lim_{z \to \lambda} 1- r(z)\prod\limits_{i=1}^m (z-p_i) = 0, \;
  \lim_{z \to \lambda} (\lambda-z)^k \prod\limits_{i=1}^m (z-p_i) = 0.
\end{align*}
As both the numerator and denominator approach zero as $z \to \lambda$, we
can evaluate the limit using L'Hospital's rule.
For any $l \in \{1,...,k-1\}$, by Lemma~\ref{lem:hard}(a),
differentiating the numerator and denominator
$l$ times and taking the limit as $z \to \lambda$ implies
\begin{align*}
  &\lim_{z \to \lambda} \f{d}{dz^l}
  \left(1- r(z) \prod\nolimits_{i=1}^m (z-p_i)\right) \\
  &= -\lim_{z \to \lambda} \f{d}{dz^l}
  \left(r(z) \prod\nolimits_{i=1}^m (z-p_i)\right)
  = 0 \\
  &\lim_{z \to \lambda} \f{d}{dz^l}
  \left((\lambda-z)^k \prod\nolimits_{i=1}^m (z-p_i)\right) = 0.
\end{align*}
Therefore, we apply L'Hospital's rule $k$ times and use
Lemma~\ref{lem:hard}(a) to obtain
\begin{align*}
  &\sum_{i=1}^m c_{p_i} (\lambda-p_i)^{-(k+1)}
  =
  \f{\lim_{z \to \lambda}\f{d}{dz^k} \left(1- r(z) \prod\limits_{i=1}^m (z-p_i)
    \right)}
 {\lim_{z \to \lambda} \f{d}{dz^k}\left((\lambda-z)^k \prod\limits_{i=1}^m (z-p_i)
      \right)}
 \\
 &= \f{-\left((-1)^{k+1} m^{(k)} (\lambda-q)^{-(m+k)} + \epsilon'\right)
   \prod\limits_{i=1}^m (\lambda-p_i)}
 {(-1)^k k! \prod\limits_{i=1}^m (\lambda-p_i)} \\
 &= \f{m^{(k)}}{k!} (\lambda-q)^{-(m+k)} - (-1)^{-k} \f{1}{k!} \epsilon' \\
 &=  {m+k-1 \choose k} (\lambda-q)^{-(m+k)} + \epsilon, \quad
 |\epsilon| \leq K\hat{d}(q).
\end{align*}
This proves Case (a).

For Case (b), we first recall the following fact, which is straightforward
to verify: \\
Fact 1. If there exist $k_{i,j}$ and $d$ positive such that
$|M_{i,j}| \leq k_{i,j} d$ for all $i$, $j$ then there exists $K > 0$
such that $||M||_2 \leq K d$.

by Fact 1 it suffices to show that for each
$l \in \{0,...,m_q-1\}$ there exists $k_l > 0$ such that the $l$th
superdiagonal of the matrix in the desired result satisfies
\begin{align*}
  \left|
  \sum\nolimits_{i=1}^m c_{p_i} (p_i-q)^m (-1)^l (\lambda-p_i)^{-(l+1)} \f{1}{z-p_i}
  \right|
  \leq k_l \hat{d}(q).
\end{align*}
By Lemma~\ref{lem:hard}(b) and since $z \in \c$,
\begin{align*}
  &\left|
  \sum\nolimits_{i=1}^m c_{p_i} (p_i-q)^m (-1)^l (\lambda-p_i)^{-(l+1)} \f{1}{z-p_i}
  \right| \\
  &
  \overset{\substack{\text{absolute} \\ \text{value}}}{=}
  \left|
  \sum\nolimits_{i=1}^m c_{p_i} (p_i-q)^m (\lambda-p_i)^{-(l+1)} \f{1}{z-p_i}
  \right| \\
  &
  \overset{\text{Lemma}~\ref{lem:hard}\text{(b)}}{=}
  \left|
  \f{(z-q)^m}{(\lambda-z)^k \prod\limits_{i=1}^m (z-p_i)}
  - \sum_{n=0}^{k-1} a_n (\lambda-z)^{n-k}
  \right| \\
  & 
  \overset{\substack{\text{triange} \\ \text{inequality}}}{\leq}
  \left|
  \f{(z-q)^m - a_0 \prod\limits_{i=1}^m (z-p_i)}{(\lambda-z)^k \prod\limits_{i=1}^m (z-p_i)}
  \right|
  + \sum_{n=1}^{k-1} |a_n| |\lambda-z|^{n-k} \\
  & 
  \overset{\substack{a_0 = 1 \\+ (a_0-1)}}{\leq}    
  \left|
  \f{(z-q)^m - \prod\limits_{i=1}^m (z-p_i) + (1-a_0) \prod\limits_{i=1}^m (z-p_i)}
    {(\lambda-z)^k \prod\limits_{i=1}^m (z-p_i)}
  \right| \\
  &+ \sum\nolimits_{n=1}^{k-1} |a_n| (1-|\lambda|)^{n-k} \\
  &
  \overset{\substack{\text{triange} \\ \text{inequality}}}{\leq}  
    \left| \f{(z-q)^m}{\prod\limits_{i=1}^m (z-p_i)} - 1 \right|  |\lambda-z|^{-k}
  + |1-a_0| |\lambda-z|^{-k} \\
  &+ \sum\nolimits_{n=1}^{k-1} |a_n| (1-|\lambda|)^{n-k} \\
  &
  \overset{z \in \c}{\leq}
  \left| \f{(z-q)^m}{\prod\limits_{i=1}^m (z-p_i)} - 1 \right| (1-|\lambda|)^{-k}
  + |1-a_0| (1-|\lambda|)^{-k} \\
  &+ \sum\nolimits_{n=1}^{k-1} |a_n| (1-|\lambda|)^{n-k}
  \leq K \hat{d}(q) \\
  K &= K'(1-|\lambda|)^{-k} + K_0'(1-|\lambda|)^{-k}
  + \sum_{n=1}^{k-1} K_n' (1-|\lambda|)^{n-k}
\end{align*}
where such $K'>0$ exists by Lemma~\ref{lem:abs}.
This proves Case (b).
\end{proof}

\begin{lemma}\label{lem:case2}
  Let $\tilde{m} \in \{0,1\}$, let $m_q, m > 0$ be integers, and let $q \in \d$.
  Suppose that for each $i \in \{1,...,m\}$ we have matrices $G_i^*$ and
  $H_i^*$, and for each $j \in \{1,...,i\}$ we have matrices $H^i_j$ and
  $G^i_j$, and poles $p^i_j$, such that the following hold:
  \begin{align}
    H^i_j = c^i_j H_i^*, \quad G^i_j = -J(q-p^i_j)^{-1}BH^i_j.
    \label{eq:hgij}
  \end{align}
  Then
  \begin{align}
  &\sum_{i=1+\tilde{m}}^m \sum_{j=1+\tilde{m}}^i G^i_j \f{1}{z-p^i_j} 
  \nonumber \\
  &+ \tilde{m} \sum_{l=2}^{m_q+1}  J(0)^{l-2} \sum_{i=1}^m BH^i_1 \f{1}{(z-q)^l}
  \nonumber \\
  &- \sum_{l=1}^{m_q} J(0)^{l-1} \sum_{i=1+\tilde{m}}^m \sum_{j=1+\tilde{m}}^i
  G^i_j \f{1}{(z-q)^l} \nonumber \\
  &= \left( \sum_{l=0}^{m_q-1} J(0)^l \tilde{m} \f{1}{(z-q)^{l+2}} \right) BH_1^*
  \nonumber \\
  &+ \sum_{i=1+\tilde{m}}^m \left(
  \sum_{j=1+\tilde{m}}^i
  J(q-p^i_j)^{-1} \f{-c^i_j}{z-p^i_j} \f{(p^i_j-q)^{m_q}}{(z-q)^{m_q}} \right.
  \nonumber \\
  &+ \sum_{l=0}^{m_q-i-1}  \sum_{k=0}^{m_q-i-1-l} J(0)^{l}
  \sum_{j=1+\tilde{m}}^i \f{c^i_j  (p^i_j-q)^{m_q-2-k-l}}{(z-q)^{m_q-k}} \nonumber \\
  &\left. \vphantom{\sum_{l=1}^2}
  + \tilde{m} c^i_1 J(0)^{m_q-1} \f{1}{(z-q)^{m_q+1}} \right) BH_i^*.
    \label{eq:big_identity}
  \end{align}
  Furthermore, for each $i \in \{1,...,m_q\}$ and $l \in \{0,...,m_q-1\}$,
  there exists $K_{i,l} > 0$ such that
  each element in the $l$th superdiagonal of the term multiplying $BH_i^*$
  in \eqref{eq:big_identity}
  has a difference from $\f{1}{(z-q)^{i+l+1}}$ bounded in absolute value by
  $K_{i,l}D(\p)$.
\end{lemma}

\begin{proof}[Proof of Lemma~\ref{lem:case2}]
We begin by proving \eqref{eq:big_identity}.
For any $i \in \{1,...m\}$, $j \in \{1,...,i\}$,
and $k \in \{1,...,m_q-1\}$ we
have that $J(0)G^i_j = -c^i_jJ(0)J(q-p^i_j)^{-1}BH_i^*$.
Writing $J(0) = J(q-p^i_j) + (p^i_j-q)I$ implies that
$J(0)G^i_j = -c^i_j BH_i^* - c^i_j (p^i_j-q) J(q-p^i_j)^{-1} BH_i^*$.
Iterating this process yields
\begin{align}
  \begin{split}
  J(0)^k G^i_j &= \sum\nolimits_{l=0}^{k-1} -c^i_j J(0)^{k-1-l} (p^i_j-q)^l BH_i^* \\
  &- c^i_j (p^i_j-q)^k J(q-p^i_j)^{-1} BH_i^*.
  \label{eq:jok}
  \end{split}
\end{align}
Then, applying Lemma~\ref{lem:easy}(a) for $k \leq i-1$, 
setting
$z = q$,
and noting that for $\tilde{m} = 1$, $c^i_j$ contains a factor of
$\f{1}{p^i_j-q}$ gives 
\begin{align}
  \sum\nolimits_{j=1+\tilde{m}}^i c^i_j  (p^i_j-q)^l = 0 \label{eq:cp0}
\end{align}
for any $i \in \{1+\tilde{m},...,m\}$ and $l \in \{1,...,i-2\}$.
Furthermore, applying Lemma~\ref{lem:easy}(a) for $k = 1$ and 
setting
$z = q$ implies
\begin{align}
  \sum\nolimits_{j=1}^i c^i_j  = 0. \label{eq:c0}
\end{align}
for any $i \in \{1+\tilde{m},...,m\}$ (and $l = 0$).
Therefore, 
for any
$i \in \{1+\tilde{m},...,m\}$ and $k \in \{1,...,m_q-1\}$
\begin{align}
  &\tilde{m} J(0)^{k-1} BH^i_1  -J(0)^k
  \sum\nolimits_{j=1+\tilde{m}}^i G^i_j \nonumber \\
  &
  \overset{\substack{\eqref{eq:jok} \\ \eqref{eq:hgij}}}{=}
  \sum\nolimits_{l=0}^{k-1} J(0)^{k-1-l} BH_i^*
  \sum\nolimits_{j=1+\tilde{m}}^i c^i_j  (p^i_j-q)^l
  \nonumber \\
  &+ \sum_{j=1+\tilde{m}}^i c^i_j (p^i_j-q)^k J(q-p^i_j)^{-1} BH_i^*
  + \tilde{m} J(0)^{k-1} c^i_1 BH_i^* \nonumber \\
  &
  \overset{\substack{\text{regroup} \\ l=0 \text{ terms}}}{=}
  \sum\nolimits_{l=1}^{k-1} J(0)^{k-1-l} BH_i^*
  \sum\nolimits_{j=1+\tilde{m}}^i c^i_j  (p^i_j-q)^l
  \nonumber \\
  &+ \sum\nolimits_{j=1+\tilde{m}}^i c^i_j (p^i_j-q)^k J(q-p^i_j)^{-1} BH_i^*
  \nonumber \\
  &+ J(0)^{k-1} BH_i^* \sum\nolimits_{j=1+\tilde{m}}^i c^i_j
  + \tilde{m} J(0)^{k-1} c^i_1 BH_i^* \nonumber \\
  &
  \overset{\substack{\text{combine last} \\ \text{two terms}}}{=}
  \sum\nolimits_{l=1}^{k-1} J(0)^{k-1-l} BH_i^*
  \sum\nolimits_{j=1+\tilde{m}}^i c^i_j  (p^i_j-q)^l
  \nonumber \\
  &+ \sum\nolimits_{j=1+\tilde{m}}^i c^i_j (p^i_j-q)^k J(q-p^i_j)^{-1} BH_i^*
  \nonumber \\
  &+ J(0)^{k-1} BH_i^* \sum\nolimits_{j=1}^i c^i_j \nonumber \\
  &
  \overset{\substack{\eqref{eq:cp0} \\ \eqref{eq:c0}}}{=}
  \sum\nolimits_{l=i-1}^{k-1} J(0)^{k-1-l} BH_i^*
  \sum\nolimits_{j=1+\tilde{m}}^i c^i_j  (p^i_j-q)^l
  \nonumber \\
  &+ \sum\nolimits_{j=1+\tilde{m}}^i c^i_j (p^i_j-q)^k J(q-p^i_j)^{-1} BH_i^*
  \nonumber \\
  &
  \overset{\eqref{eq:gij_true}}{=}
  \sum\nolimits_{l=i-1}^{k-1} J(0)^{k-1-l} BH_i^* \sum\nolimits_{j=1+\tilde{m}}^i c^i_j
  (p^i_j-q)^l
  \nonumber \\
  &- \sum\nolimits_{j=1+\tilde{m}}^i (p^i_j-q)^k G^i_j. \label{eq:gij}
\end{align}
We compute
\begin{align}
  &\Phi_x(z) 
  \overset{\substack{\eqref{eq:gl} \\ \eqref{eq:phixt}}}{=}
  \sum\nolimits_{i=1+\tilde{m}}^m \sum\nolimits_{j=1+\tilde{m}}^i G^i_j \f{1}{z-p^i_j}
  \nonumber\\
  &+ \tilde{m} \sum\nolimits_{l=2}^{m_q+1}  J(0)^{l-2}
  \sum\nolimits_{i=1}^m BH^i_1 \f{1}{(z-q)^l}
  \nonumber \\
  &- \sum\nolimits_{l=1}^{m_q} J(0)^{l-1} \sum\nolimits_{i=1+\tilde{m}}^m
  \sum\nolimits_{j=1+\tilde{m}}^i
  G^i_j \f{1}{(z-q)^l} \nonumber \\
  &
  \overset{\substack{\text{combining} \\ \text{terms in sum}}}{=}
  \sum\nolimits_{i=1+\tilde{m}}^m \sum\nolimits_{j=1+\tilde{m}}^i G^i_j \f{1}{z-p^i_j}
  \nonumber \\
  &+  \sum_{l=2}^{m_q} \sum_{i=1+\tilde{m}}^m \left(\tilde{m} J(0)^{l-2} BH^i_1
  - J(0)^{l-1} \mkern-15mu \sum_{j=1+\tilde{m}}^i G^i_j  \right)
  \f{1}{(z-q)^l} \nonumber \\
  &- \sum\nolimits_{i=1+\tilde{m}}^m \sum\nolimits_{j=1+\tilde{m}}^i G^i_j \f{1}{z-q}
  \nonumber \\
 &+ \tilde{m} J(0)^{m_q-1} \sum\nolimits_{i=1+\tilde{m}}^m BH^i_1 \f{1}{(z-q)^{m_q+1}}
  \nonumber \\
  &+ \sum\nolimits_{l=0}^{m_q-1} \tilde{m} J(0)^l BH^1_1 \f{1}{(z-q)^{l+2}}
  \nonumber \\
  &
  \overset{\eqref{eq:gij}}{=}
  \sum\nolimits_{i=1+\tilde{m}}^m \sum\nolimits_{j=1+\tilde{m}}^i G^i_j \f{1}{z-p^i_j}
  \nonumber \\
  &+  \sum_{l=2}^{m_q} \sum_{i=1+\tilde{m}}^m \sum_{k=i-1}^{l-2} J(0)^{l-2-k} BH_i^*
  \sum_{j=1+\tilde{m}}^i \f{c^i_j  (p^i_j-q)^k }{(z-q)^l} \nonumber \\
  &- \sum\nolimits_{l=2}^{m_q} \sum\nolimits_{i=1+\tilde{m}}^m
  \sum\nolimits_{j=1+\tilde{m}}^i (p^i_j-q)^{l-1} G^i_j \f{1}{(z-q)^l} \nonumber \\
  &- \sum\nolimits_{i=1+\tilde{m}}^m \sum\nolimits_{j=1+\tilde{m}}^i G^i_j \f{1}{z-q}
  \nonumber \\
 &+ \tilde{m} J(0)^{m_q-1} \sum\nolimits_{i=1+\tilde{m}}^m BH^i_1 \f{1}{(z-q)^{m_q+1}}
  \nonumber\\
  &+ \sum\nolimits_{l=0}^{m_q-1} \tilde{m} J(0)^l BH^1_1 \f{1}{(z-q)^{l+2}}.
  \label{eq:temp}
\end{align}
Note that 
\begin{align}
  &\sum\nolimits_{i=1+\tilde{m}}^m \sum\nolimits_{j=1+\tilde{m}}^i
  G^i_j \f{1}{z-p^i_j} \nonumber \\
  &- \sum\nolimits_{l=2}^{m_q} \sum\nolimits_{i=1+\tilde{m}}^m
  \sum\nolimits_{j=1+\tilde{m}}^i (p^i_j-q)^{l-1} G^i_j \f{1}{(z-q)^l} \nonumber \\
  &- \sum\nolimits_{i=1+\tilde{m}}^m \sum\nolimits_{j=1+\tilde{m}}^i G^i_j \f{1}{z-q}
  \nonumber \\
  &
  \overset{\text{regrouping}}{=}
  \sum\nolimits_{i=1+\tilde{m}}^m \sum\nolimits_{j=1+\tilde{m}}^i G^i_j
  \left(\f{1}{z-p^i_j} - \f{1}{z-q}\right) \nonumber \\
  &- \sum\nolimits_{l=2}^{m_q} \sum\nolimits_{i=1+\tilde{m}}^m
  \sum\nolimits_{j=1+\tilde{m}}^i (p^i_j-q)^{l-1} G^i_j \f{1}{(z-q)^l} \nonumber \\
  &
  \overset{\text{simplifying}}{=}
  \sum\nolimits_{i=1+\tilde{m}}^m \sum\nolimits_{j=1+\tilde{m}}^i (p^i_j-q) G^i_j
  \f{1}{z-q}\f{1}{z-p^i_j} \nonumber \\
  &- \sum\nolimits_{l=2}^{m_q} \sum\nolimits_{i=1+\tilde{m}}^m
  \sum\nolimits_{j=1+\tilde{m}}^i (p^i_j-q)^{l-1} G^i_j \f{1}{(z-q)^l} \nonumber \\
  &
  \overset{\text{regrouping}}{=}
  \sum_{i=1+\tilde{m}}^m \sum_{j=1+\tilde{m}}^i (p^i_j-q) G^i_j
  \f{1}{z-q}\left(\f{1}{z-p^i_j} - \f{1}{z-q}\right) \nonumber \\
  &- \sum\nolimits_{l=3}^{m_q} \sum\nolimits_{i=1+\tilde{m}}^m
  \sum\nolimits_{j=1+\tilde{m}}^i (p^i_j-q)^{l-1} G^i_j \f{1}{(z-q)^l} \nonumber \\
  &
  \overset{\text{simplifying}}{=}
  \sum\nolimits_{i=1+\tilde{m}}^m \sum\nolimits_{j=1+\tilde{m}}^i (p^i_j-q)^2 G^i_j
  \f{1}{(z-q)^2} \f{1}{z-p^i_j} \nonumber \\
  &- \sum\nolimits_{l=3}^{m_q} \sum\nolimits_{i=1+\tilde{m}}^m
  \sum\nolimits_{j=1+\tilde{m}}^i (p^i_j-q)^{l-1} G^i_j \f{1}{(z-q)^l} \nonumber \\
  & ... \nonumber \\
  &
  \overset{\substack{\text{iterating} \\ \text{the above}}}{=}
  \sum\nolimits_{i=1+\tilde{m}}^m
  \sum\nolimits_{j=1+\tilde{m}}^i (p^i_j-q)^{m_q} G^i_j \f{1}{z-p^i_j}
  \f{1}{(z-q)^{m_q}} \nonumber \\
  &
  \overset{\eqref{eq:gij_true}}{=}
  \sum_{i=1+\tilde{m}}^m \sum_{j=1+\tilde{m}}^i J(q-p^i_j)^{-1}BH_i^*
  \f{-c^i_j}{z-p^i_j} \f{(p^i_j-q)^{m_q}}{(z-q)^{m_q}}
  \label{eq:first}
\end{align}
Also, 
we compute
\begin{align}
  &\sum_{l=2}^{m_q} \sum_{i=1+\tilde{m}}^m \sum_{k=i-1}^{l-2} J(0)^{l-2-k} BH_i^*
  \sum_{j=1+\tilde{m}}^i \f{c^i_j  (p^i_j-q)^k }{(z-q)^l} \nonumber \\
  &
  \overset{\substack{\text{reverse} \\ \text{sums}}}{=}
  \sum_{i=1+\tilde{m}}^m \sum_{l=i+1}^{m_q}  \sum_{k=i-1}^{l-2} J(0)^{l-2-k} BH_i^*
  \sum_{j=1+\tilde{m}}^i \f{c^i_j  (p^i_j-q)^k }{(z-q)^l} \nonumber \\
  %
  %
  %
  %
  %
  &
  \overset{\substack{l' = l-2-k \\ k' = m_q-l}}{=}
  \sum\nolimits_{i=1+\tilde{m}}^m \sum\nolimits_{k'=0}^{m_q-i-1}
  \sum\nolimits_{l'=0}^{m_q-i-1-k'}
  \nonumber \\
  & \mkern+100mu J(0)^{l'} BH_i^*
  \sum\nolimits_{j=1+\tilde{m}}^i \f{c^i_j  (p^i_j-q)^{m_q-2-k'-l'}}{(z-q)^{m_q-k'}}
  \nonumber \\
  %
  %
  &
  \overset{\substack{\text{reverse} \\ \text{sums}}}{=}
  \sum\nolimits_{i=1+\tilde{m}}^m \sum\nolimits_{l'=0}^{m_q-i-1}
  \sum\nolimits_{k'=0}^{m_q-i-1-l'}
  \nonumber \\
  & \mkern+100mu J(0)^{l'} BH_i^*
  \sum\nolimits_{j=1+\tilde{m}}^i \f{c^i_j  (p^i_j-q)^{m_q-2-k'-l'}}{(z-q)^{m_q-k'}}.
  \label{eq:second}
\end{align}
Furthermore, we have
\begin{align}
  \begin{split}
   &\tilde{m} J(0)^{m_q-1} \sum\nolimits_{i=1+\tilde{m}}^m BH^i_1 \f{1}{(z-q)^{m_q+1}}
    \\
    & \mkern+80mu
    \overset{\eqref{eq:hgij}}{=}
    \tilde{m} \sum\nolimits_{i=1+\tilde{m}}^m c^i_1 J(0)^{m_q-1} BH_i^*
    \f{1}{(z-q)^{m_q+1}}
  \\
  &\sum\nolimits_{l=0}^{m_q-1} \tilde{m} J(0)^l BH^1_1 \f{1}{(z-q)^{l+2}} \\
  & \mkern+80mu
  \overset{\eqref{eq:hgij}}{=}
  \sum\nolimits_{l=0}^{m_q-1} \tilde{m} J(0)^l BH_1^* \f{1}{(z-q)^{l+2}}.
  \end{split}
  \label{eq:third}
\end{align}
Substituting \eqref{eq:first}, \eqref{eq:second}, and \eqref{eq:third} into
\eqref{eq:temp} yields
\begin{align*}
  &\Phi_x(z)
  = \left( \sum\nolimits_{l=0}^{m_q-1} J(0)^l \tilde{m} \f{1}{(z-q)^{l+2}} \right)
  BH_1^*
  \\
  &+ \sum\nolimits_{i=1+\tilde{m}}^m \left(
  \sum\nolimits_{j=1+\tilde{m}}^i
  J(q-p^i_j)^{-1} \f{-c^i_j}{z-p^i_j} \f{(p^i_j-q)^{m_q}}{(z-q)^{m_q}} \right.
  \\
  &+ \sum_{l=0}^{m_q-i-1}  \sum_{k=0}^{m_q-i-1-l} J(0)^{l}
  \sum_{j=1+\tilde{m}}^i \f{c^i_j  (p^i_j-q)^{m_q-2-k-l}}{(z-q)^{m_q-k}} \\
  &\left. \vphantom{\sum_{l=1}^2}
  + \tilde{m} c^i_1 J(0)^{m_q-1} \f{1}{(z-q)^{m_q+1}} \right) BH_i^*.
\end{align*}
This completes the proof of \eqref{eq:big_identity}.

Next we derive the upper bound of the difference from $\f{1}{(z-q)^{i+l+1}}$
for the elements of the superdiagonal of the
term multiplying $BH_i^*$.  Note that, by the form of \eqref{eq:big_identity},
all elements on each such superdiagonal are identical.
By \eqref{eq:big_identity}, for $i \in \{1+\tilde{m},...,m\}$
and $l \in \{m_q-i,...,m_q-2\}$,
the $l$th superdiagonal
of the term multiplying $BH_i^*$ in $\Phi_x(z)$ is 
\begin{align*}
  &\left( \sum\nolimits_{j=1+\tilde{m}}^i
   (-1)^l (q-p^i_j)^{-(l+1)} \f{-c^i_j (p^i_j-q)^{m_q}}{z-p^i_j} \right)
  \f{1}{(z-q)^{m_q}} \\
  &=   \left( \sum_{j=1+\tilde{m}}^i
   \f{(-1)^{l+1}}{(-1)^{l+1}} (p^i_j-q)^{-(l+1)}
  \f{c^i_j (p^i_j-q)^{m_q}}{z-p^i_j} \right) \f{1}{(z-q)^{m_q}} \\
  &= \left( \sum\nolimits_{j=1+\tilde{m}}^i c^i_j (p^i_j-q)^{m_q-1-l} \f{1}{z-p^i_j}
  \right)
  \f{1}{(z-q)^{m_q}} \\
  &
  \overset{\text{Lemma}~\ref{lem:easy}\text{(a)}}{=}
  \f{(z-q)^{m_q-1-l-\tilde{m}}}{\prod\nolimits_{j=1+\tilde{m}}^i (z-p^i_j)}
  \f{1}{(z-q)^{m_q}}
  = \f{(z-q)^{-(l+1+\tilde{m})}}{\prod\nolimits_{j=1+\tilde{m}}^i (z-p^i_j)}
\end{align*}
where for Lemma~\ref{lem:easy}(a) note that for
$\tilde{m} = 1$, $c^i_j$ contains a factor of $\f{1}{p^i_j-q}$.
For $i \in \{1+\tilde{m},...,m\}$,
the $(m_q-1)$th superdiagonal  (i.e. $l = m_q-1$)
of the term multiplying $BH_i^*$ in $\Phi_x(z)$ is 
\begin{align*}
  &\left( \sum_{j=1+\tilde{m}}^i
   (-1)^{m_q-1} (q-p^i_j)^{-m_q} \f{-c^i_j (p^i_j-q)^{m_q}}{z-p^i_j} \right)
  \f{1}{(z-q)^{m_q}}
  \\
  &+ \tilde{m} c^i_1 \f{1}{(z-q)^{m_q+1}} \\
  &= \left( \sum_{j=1+\tilde{m}}^i
   \f{(-1)^{m_q}}{(-1)^{m_q}} (p^i_j-q)^{-m_q} \f{c^i_j (p^i_j-q)^{m_q}}{z-p^i_j}
  \right)
  \f{1}{(z-q)^{m_q}} \\
  &+ \tilde{m} c^i_1 \f{1}{(z-q)^{m_q+1}} \\
  &= \left( \sum\nolimits_{j=1+\tilde{m}}^i c^i_j \f{1}{z-p^i_j} \right)
  \f{1}{(z-q)^{m_q}}
  + \tilde{m} c^i_1 \f{1}{(z-q)^{m_q+1}} \\
  &= \left( \sum\nolimits_{j=1+\tilde{m}}^i c^i_j \f{1}{z-p^i_j}
  + \tilde{m} c^i_1 \f{1}{(z-q)} \right) \f{1}{(z-q)^{m_q}} \\
  &= \left( \sum_{j=1}^i c^i_j \f{1}{z-p^i_j} \right) \f{1}{(z-q)^{m_q}}
  \overset{\text{[1, Eq.~7]}}{=} 
  \mkern-10mu
  \f{1}{\prod\limits_{j=1}^i (z-p^i_j)} \f{1}{(z-q)^{m_q}} \\
  &= \f{1}{\prod\nolimits_{j=1}^i (z-p^i_j)} \f{1}{(z-q)^{l+1}}.
\end{align*}
For $i \in \{1+\tilde{m},...,m\}$ and $l \in \{0,...,m_q-i-1\}$,
the $l$th superdiagonal
of the term multiplying $BH_i^*$ in $\Phi_x(z)$ is
\begin{align*}
  &\left( \sum\nolimits_{j=1+\tilde{m}}^i
   (-1)^l (q-p^i_j)^{-(l+1)}\f{-c^i_j (p^i_j-q)^{m_q}}{z-p^i_j} \right)
  \f{1}{(z-q)^{m_q}} \\
  &+ \sum\nolimits_{k=0}^{m_q-i-1-l} 
  \sum\nolimits_{j=1+\tilde{m}}^i c^i_j  (p^i_j-q)^{m_q-2-k-l} \f{1}{(z-q)^{m_q-k}} \\
  &=  \left( \sum\nolimits_{j=1+\tilde{m}}^i c^i_j (p^i_j-q)^{m_q-l-1}
  \f{1}{z-p^i_j} \right)
  \f{1}{(z-q)^{m_q}} \\
  &+ \sum\nolimits_{k=0}^{m_q-i-1-l} 
  \sum\nolimits_{j=1+\tilde{m}}^i c^i_j  (p^i_j-q)^{m_q-2-k-l} \f{1}{(z-q)^{m_q-k}} \\
  &
  \overset{\text{Lemma}~\ref{lem:easy}\text{(b)}}{=}
  \f{1}{(z-q)^{m_q}} \left(
 \f{(z-q)^{m_q-l-1-\tilde{m}}}{\prod\nolimits_{j=1+\tilde{m}}^i (z-p^i_j)} \right. \\
 &\left.
 \mkern+20mu 
 - \sum\nolimits_{k=0}^{m_q-i-1-l}
 \sum\nolimits_{j=1+\tilde{m}}^i c^i_j (p^i_j-q)^{m_q-2-k-l} (z-q)^k
 \right)   \\
  &+ \sum\nolimits_{k=0}^{m_q-i-1-l} 
  \sum\nolimits_{j=1+\tilde{m}}^i c^i_j  (p^i_j-q)^{m_q-2-k-l} \f{1}{(z-q)^{m_q-k}} \\
  &
  \overset{\text{simplifying}}{=}
  \f{1}{\prod\nolimits_{j=1+\tilde{m}}^i (z-p^i_j)} \f{1}{(z-q)^{l+1+\tilde{m}}} \\
  &- \sum\nolimits_{k=0}^{m_q-i-1-l}
  \sum\nolimits_{j=1+\tilde{m}}^i c^i_j (p^i_j-q)^{m_q-2-k-l} 
  \f{1}{(z-q)^{m_q-k}} \\
  &+ \sum\nolimits_{k=0}^{m_q-i-1-l} 
  \sum\nolimits_{j=1+\tilde{m}}^i c^i_j  (p^i_j-q)^{m_q-2-k-l} \f{1}{(z-q)^{m_q-k}} \\
  &
  \overset{\text{cancelling}}{=}
  \f{1}{\prod\nolimits_{j=1+\tilde{m}}^i (z-p^i_j)} \f{1}{(z-q)^{l+1+\tilde{m}}}
\end{align*}
where for Lemma~\ref{lem:easy}(b) note that for
$\tilde{m} = 1$, $c^i_j$ contains a factor of $\f{1}{p^i_j-q}$.
Finally, if $\tilde{m} = 1$, then for $i = 1$ and any $l \in \{0,...,m_q-1\}$,
the $l$th superdiagonal of the term multiplying $BH_1^*$ in $\Phi_x(z)$ is
given by
\begin{align*}
  \tilde{m} \f{1}{(z-q)^{l+2}} = \tilde{m} \f{1}{(z-q)^{i+l+1}}.
\end{align*}
Thus, combining the cases above, by Lemma~\ref{lem:abs} for every
$i \in \{1,...,m\}$
and $l \in \{0,...,m_q-1\}$, each term in the $l$th superdiagonal of the
term multiplying $BH_i^*$ in \eqref{eq:big_identity} has difference from
$\f{1}{(z-q)^{i+l+1}}$ bounded by $K_{i,l} D(\p)$ for some $K_{i,l} > 0$.
\end{proof}

Now we are ready to prove Lemma~\ref{lem:bound}.

\begin{proof}[Proof of Lemma~\ref{lem:bound}]
The proof begins by selecting an optimal solution ($\Phi_x^*$,$\Phi_u^*$) to
the infinite dimensional control design problem of \eqref{eq:spa_inf},
and constructing $\Phi_u(z) = \sum_{p \in \p} H_p \f{1}{z-p}$ by
\cite[Theorem~1]{Fi22a} to approximate $\Phi_u^*$.
By \cite[Theorem~1]{Fi22a}, this immediately implies that the
approximation error bounds for $\Phi_u$ of \eqref{eq:ubound} are satisfied.
Next, $\Phi_x$ is defined as the unique transfer function that satisfies
the SLS constraint in \eqref{eq:spa_inf}.
The remainder of the proof will show that $\Phi_x$ is a feasible solution to
\eqref{eq:spa_obj}-\eqref{eq:spa_imp}, and
that it satisfies the approximation error bounds of \eqref{eq:xbound}.

Towards that end, first it is shown that it suffices to work in coordinates in
which $A$ is in Jordan normal form.
Next it is shown that, in these coordinates, the approximation error bounds
and the SLS constraints decouple according to each elementary Jordan block in
$A$, so it suffices to prove the result for a single elementary Jordan block
with eigenvalue $\lambda$.
Afterwards, it is shown that the SLS constraint uniquely determines the poles
and multiplicities of $\tilde{\Phi}_x$ from those of $\tilde{\Phi}_u$
for any transfer functions $(\tilde{\Phi}_x,\tilde{\Phi}_u)$ in $\f{1}{z}\h$
that satisfy it.
From the choice of $\Phi_u$, this immediately implies that $\Phi_x$ is a
feasible solution to \eqref{eq:spa_obj}-\eqref{eq:spa_imp}.

Subsequently, for each pole $q$ in $\Phi_x^*$ that appears in $\Phi_u^*$,
by \cite[Theorem~1]{Fi22a} there exist poles in $\Phi_u$ for approximating
the portion of $\Phi_u^*$ corresponding to pole $q$.  By the relationship
between $\Phi_u$ and $\Phi_x$ described above, we then consider the resulting
poles that appear in $\Phi_x$, and will show that the portion of $\Phi_x$
corresponding to these poles closely approximates the portion of $\Phi_x^*$
corresponding to the pole $q$.  To do so, we fix a pole $q$ in $\Phi_x^*$
and consider two cases: Case 1 where $q \neq \lambda$, and Case 2 where
$q = \lambda$.
For each of these cases we use the SLS constraints to determine the
coefficients in the portions of $\Phi_x^*$ and $\Phi_x$ corresponding to pole
$q$ and the poles used to approximate it, respectively, and then bound the
resulting approximation error.
As $q$ was arbitrary, this then yields the desired approximation error bounds
for $\Phi_x$ of \eqref{eq:xbound}.

First we obtain an optimal solution to the infinite dimensional control design
problem, and use \cite[Theorem~1]{Fi22a} to find $\Phi_u$ which closely
approximates $\Phi_u^*$.
Let $(\Phi_x^*,\Phi_u^*)$ be an optimal solution to \eqref{eq:spa_inf}, which
exists by Assumption A6.
By \cite[Theorem~1]{Fi22a}, there exist coefficient matrices
$\{H_p\}_{p \in \p}$ such that, if we define
$\Phi_u(z) = \sum_{p \in \p} H_p \f{1}{z-p}$ then $\Phi_u \in \f{1}{z}\h$,
$\left|\left|\Phi_u-\Phi_u^*\right|\right|_{H_\infty} \leq K^u_\infty D(\p)$,
and $\left|\left|\Phi_u-\Phi_u^*\right|\right|_{H_2} \leq K^u_2 D(\p)$.
Define $\Phi_x(z) = (zI-A)^{-1}(B\Phi_u(z)+I)$ and note that this implies
$(\Phi_x,\Phi_u)$ satisfy the SLS constraint in \eqref{eq:spa_inf} by
construction.

As $\mathcal{Q}$ and $\sigma$ are finite,
$\eta = \min_{q \in \mathcal{Q}, \lambda \in \sigma, \lambda \neq \sigma}
d(\lambda,q) > 0$
and $d(\lambda,q) \geq \eta$ for all such $\lambda \neq q$.
By Assumption A6, for every $q \in \mathcal{Q}$ and $\lambda \in \sigma$
with $\lambda \neq q$, $d(\lambda,\p(q)) > 0$, where $\p(q)$ are the $m_q$
closest poles in $\p$ to $q$.
This implies that $\delta = \min_{q \in \mathcal{Q}, \lambda \in \sigma, \lambda \neq q}
d(\lambda,\p(q)) > 0$
and that $d(\lambda,\p(q)) \geq \delta$ for all such $\lambda \neq q$.

Next we show that it suffices to work in Jordan normal form, and with a
single elementary Jordan block.
There exists matrices $J$ in Jordan normal form and $V$ invertible such that
$J = VAV^{-1}$.
Fix $z \in \c$ for the remainder of the proof.
We will show that there exists $K > 0$ such that
\begin{align}
  ||V\Phi_x(z)-V\Phi_x^*(z)||_2 \leq K D(\p). \label{eq:goal_phix}
\end{align}
This will imply that
\begin{align*}
  &||\Phi_x-\Phi_x^*||_{H_\infty} =
  \sup_{z \in \c} ||V^{-1}(V\Phi_x(z)-V\Phi_x^*(z))||_2 \\
  &\leq ||V^{-1}||_2 \sup_{z \in \c} ||V\Phi_x(z)-V\Phi_x^*(z)||_2 
  \leq K^x_\infty D(\p) \\
  &||\Phi_x-\Phi_x^*||_{H_2} \leq \sqrt{n}||\Phi_x-\Phi_x^*||_{H_\infty}
  \leq K^x_2 D(\p) \\
  & K^x_\infty = ||V^{-1}||_2 K, \quad K^x_2 = \sqrt{n} K^x_\infty.
\end{align*}
So, to prove the lemma it suffices to show that \eqref{eq:goal_phix} holds
and that $(\Phi_x,\Phi_u)$ is a feasible solution to
\eqref{eq:spa_obj}-\eqref{eq:spa_imp}.
Furthermore, letting $J(\lambda)$ denote an elementary Jordan block with
eigenvalue $\lambda$ in $J$, and $M|_{J(\lambda)}$ the restriction of the
matrix $M$ to the rows corresponding to the rows of $J(\lambda)$ in $J$, we have
\begin{align*}
  &\left|\left| V\Phi_x(z) - V\Phi_x^*(z) \right|\right|_2 \\
  &\leq \sum\nolimits_{\lambda \in \sigma} \sum\nolimits_{J(\lambda) \text{ in } J}
  \left|\left|(V\Phi_x(z))|_{J(\lambda)} - (V\Phi_x^*(z))|_{J(\lambda)}
  \right|\right|_2.
\end{align*}
Thus, to prove \eqref{eq:goal_phix} it suffices to show that for each
elementary Jordan block $J(\lambda)$ in $J$ there exists a constant
$K_{J(\lambda)} > 0$ such that
\begin{align}
    \left|\left|(V\Phi_x(z))|_{J(\lambda)} - (V\Phi_x^*(z))|_{J(\lambda)}
    \right|\right|_2 \leq K_{J(\lambda)} D(\p).
    \label{eq:goal_block}
\end{align}
Premultiplying the SLS constraint in \eqref{eq:spa_inf} by $V$ implies
that $(zI-J)V\Phi_x - VB\Phi_u = V$, and note that this is satisfied by both
$(V\Phi_x,\Phi_u)$ and $(V\Phi_x^*,\Phi_u^*)$.
As $(zI-J)$ is block diagonal, this equation decouples into
independent equations for each elementary Jordan block $J(\lambda)$ in $J$
given by
$(zI-J(\lambda))(V\Phi_x(z))|_{J(\lambda)} - (VB)|_{J(\lambda)} \Phi_u(z) =
V|_{J(\lambda)}$.
Therefore, both our objective \eqref{eq:goal_block} and the SLS constraints
become decoupled for each $J(\lambda)$, so for the remainder of the proof
we fix a particular $\lambda \in \sigma$ and $J(\lambda)$ in $J$.
For notational convenience, for the remainder of the proof we abuse notation
and let $\Phi_x$, $\Phi_x^*$, $B$, and $V$ denote
$(V\Phi_x)|_{J(\lambda)}$, $(V\Phi_x^*)|_{J(\lambda)}$, $(VB)|_{J(\lambda)}$,
and $V|_{J(\lambda)}$, respectively.
Then the objective \eqref{eq:goal_block} and the SLS constraint become
\begin{align}
  \left|\left|\Phi_x(z) - \Phi_x^*(z) \right|\right|_2 &\leq K D(\p)
  \label{eq:goal_block2}
  \\
    (zI-J(\lambda)) \Phi_x(z) - B \Phi_u(z) &= V. \label{eq:final_cons}
\end{align}
To complete the proof it suffices to show that there exists $K > 0$ such
that \eqref{eq:goal_block2} holds, and that
$(\Phi_x,\Phi_u)$ is a feasible solution to
\eqref{eq:spa_obj}-\eqref{eq:spa_imp}.

Let $m_\lambda$ denote the multiplicity of $\lambda$ in $J(\lambda)$.
As $(zI-A)^{-1}$ is strictly proper real rational and $(B\Phi_u(z)+I)$ is
proper real rational,
their product $\Phi_x$ is strictly proper real rational.
Therefore, $\Phi_x$ has a partial fraction decomposition which does not
include any constant or polynomial terms, and in which all poles have
finite multiplicity.

Now we derive the relationship between the poles and multiplicities of
any pair of transfer functions which satisfy the SLS constraint.
Let $(\tilde{\Phi}_x,\tilde{\Phi}_u)$ be any transfer functions which
satisfy \eqref{eq:final_cons} and such that $\tilde{\Phi}_u \in \f{1}{z} \h$
and $\tilde{\Phi}_x$ is strictly proper rational.
Let $q$ be any pole of $\tilde{\Phi}_x$, $m_q = m_\lambda$ if
$q = \lambda$ or $m_q = 0$ otherwise, and $m$ the multiplicity of
$q$ in $\tilde{\Phi}_u$.
Note that $m = 0$ if $q$ is not a pole of $\tilde{\Phi}_u$.
Let $\hat{m}$ be the multiplicity of $q$ in $\tilde{\Phi}_x$.
Then the terms in the partial fraction decompositions of $\tilde{\Phi}_u$ and
$\tilde{\Phi}_x$ corresponding to pole $q$ are given by
$\sum_{i=1}^m H^i \f{1}{(z-q)^i}$ and
$\sum_{i=1}^{\hat{m}} G^i \f{1}{(z-q)^i}$, respectively.
By uniqueness of the partial fraction decomposition, \eqref{eq:final_cons}
therefore implies that
\begin{align}
  G^{i+1} &= J(\lambda-q) G^i + BH^i, \quad i \in \{1,...,m\} \label{eq:bh} \\
  G^{i+1} &= J(\lambda-q) G^i, \quad i \in \{m+1,...,\hat{m}-1\}
  \label{eq:ignore} \\
  0 &= J(\lambda-q)G^{\hat{m}}. \label{eq:zero}
\end{align}
First consider the case where $\lambda \neq q$.
It is straightforward to verify the following fact: \\
Fact 2. If $J(\lambda)G = 0$ for $\lambda \neq 0$ then $G = 0$.

Then by Fact 2 and \eqref{eq:zero}, $G^{\hat{m}} = 0$.
Proceeding downwards in $i$, repeated application of Fact 2 and
\eqref{eq:ignore} imply that $G^i = 0$ for $i \in \{m+1,...,\hat{m}\}$.
So, in this case the order of $q$ in $\tilde{\Phi}_x$ is $m = m + m_q$.
Next consider the case where $\lambda = q$.
Then, by \eqref{eq:ignore}, $G^i = J(0)^{i-(m+1)}G^{m+1}$ for
$i \in \{m+1,...,\hat{m}-1\}$.
As $J(0)^{m_\lambda} = 0 = J(0)^{m_q}$, this implies that
$G^i = 0$ for $i \geq m_q + m + 1$.
So, in this case the order of $q$ in $\tilde{\Phi}_x$ is $m+m_q$.
Combining the above cases implies the following fact:
that $\tilde{\Phi}_x$ only contains poles in
$\sigma$ and $\tilde{\Phi}_u$, and that their multiplicities are given by
$m+m_q$.
Applying this fact to $(\Phi_x,\Phi_u)$ implies that $\Phi_x \in \f{1}{z} \h$
and is a feasible solution to \eqref{eq:spa_obj}-\eqref{eq:spa_imp}.
Hence, to complete the proof it suffices to prove \eqref{eq:goal_block2}.

In what follows, we show that to prove \eqref{eq:xbound} it suffices to fix
a particular pole in $\Phi_x^*$, and to show that a certain portion of
$\Phi_x$ closely approximates the portion of $\Phi_x^*$ corresponding to this
pole.  This is done by using the construction of \cite[Theorem~1]{Fi22a} to
approximate $\Phi_u^*$ by $\Phi_u$.
Let $\mathcal{Q}$ denote the poles of $\Phi_x^*$.
For each $q \in \mathcal{Q}$, its contribution to the partial fraction
decompositions of $\Phi_x^*$ and $\Phi_u^*$ is given, respectively, by
\begin{align}
  \sum\nolimits_{i=1}^{m_q+m} G_i^* \f{1}{(z-q)^i}, \quad
  \sum\nolimits_{i=1}^m H_i^* \f{1}{(z-q)^i} \label{eq:phiux}
\end{align}
by the above fact.
Since $\Phi_u$ was constructed as in \cite[Theorem~1]{Fi22a}, the portion of
$\Phi_u$ that was chosen to approximate the pole at $q$ in $\Phi_u^*$ is given
by
\begin{align}
  \sum\nolimits_{i=1}^m \sum\nolimits_{j=1}^i H^i_j \f{1}{z-p^i_j},
  \quad 
  H^i_j = c^i_j H_i^* \label{eq:phiu} 
\end{align}
for all $i \in \{1,...,m\}$ and $j \in \{1,...,i\}$, where
$\{c^i_j\}_{j=1}^i$ are the constants chosen as in \cite[Corollary~2]{Fi22a}
for approximating the pole $q$ with multiplicity $i$ by the poles
$\{p^i_j\}_{j=1}^i$.
Let $\tilde{m} = 1$ if $q = \lambda$ and $\Phi_u$ contains a pole at $q$, and
$\tilde{m} = 0$ otherwise.
If $\tilde{m} = 1$, reorder the poles in $\{p^i_j\}_{j=1}^i$ for each
$i \in \{1,...,m\}$ such that $p^i_1 = q$.
Then the above fact implies that the portion of $\Phi_x$ corresponding to
the above portion of $\Phi_u$ is given by
\begin{align}
  \sum\nolimits_{i=1}^{m_q+\tilde{m}} G_i \f{1}{(z-q)^i}
  + \sum\nolimits_{i=1+\tilde{m}}^m \sum\nolimits_{j=1+\tilde{m}}^m
  G^i_j \f{1}{(z-p^i_j)}.
  \label{eq:phix}
\end{align}
Hence, from \eqref{eq:phix} and \eqref{eq:phiux} we compute
\begin{align*}
  &\left|\left|
  \Phi_x(z) - \Phi_x^*(z)
  \right|\right|_2 \\
  &=   \left|\left|
  \sum_{q \in \mathcal{Q}}   \sum_{i=1}^{m_q+\tilde{m}} G_i \f{1}{(z-q)^i}
  + \sum_{i=1+\tilde{m}}^m \sum_{j=1+\tilde{m}}^m G^i_j \f{1}{(z-p^i_j)}
  \right.\right.\\
  &\quad\;\; \left.\left.- \sum\nolimits_{q \in \mathcal{Q}}
  \sum\nolimits_{i=1}^{m_q+m} G_i^*
  \f{1}{(z-q)^i}
  \right|\right|_2 \\
  &\leq \sum_{q \in \mathcal{Q}}
  \left|\left|
  \sum_{i=1}^{m_q+\tilde{m}} G_i \f{1}{(z-q)^i}
  + \sum_{i=1+\tilde{m}}^m \sum_{j=1+\tilde{m}}^m G^i_j \f{1}{(z-p^i_j)}
  \right.\right. \\
  &\quad\quad\quad\; \left.\left.
  - \sum\nolimits_{i=1}^{m_q+m} G_i^* \f{1}{(z-q)^i}
  \right|\right|_2
\end{align*}
so, since $\mathcal{Q}$ is finite, to prove \eqref{eq:goal_block2} it suffices
to show that there exists $K_q > 0$ such that
\begin{align}
  \begin{split}
    &\left|\left|
  \sum\nolimits_{i=1}^{m_q+\tilde{m}} G_i \f{1}{(z-q)^i}
  + \sum\nolimits_{i=1+\tilde{m}}^m \sum\nolimits_{j=1+\tilde{m}}^m G^i_j \f{1}{(z-p^i_j)}
  \right.\right. \\
  &\;\left.\left.- \sum\nolimits_{i=1}^{m_q+m} G_i^* \f{1}{(z-q)^i}
  \right|\right|
  \leq K_q D(\p)
  \end{split}
  \label{eq:new_goal}
\end{align}
for each $q \in \mathcal{Q}$.
Towards that end, fix $q \in \mathcal{Q}$
and for the remainder of the proof let $\Phi_x(z)$ and $\Phi_x^*(z)$ denote the
contributions to $\Phi_x(z)$ and $\Phi_x^*(z)$ given by \eqref{eq:phix}
and \eqref{eq:phiux}, respectively, as in \eqref{eq:new_goal}.
We consider two cases.

Case 1: $q \neq \lambda$.
Substituting \eqref{eq:phix}, \eqref{eq:phiu},
and \eqref{eq:phiux} into \eqref{eq:bh}-\eqref{eq:zero} implies that
\begin{align*}
  -J(\lambda-q)G_m^* &= BH_m^* \\
  G_{i-1}^* &= J(\lambda-q)^{-1}(G_i^* - BH_{i-1}^*), \quad i \in \{2,...,m\} \\
  -J(\lambda-p^i_j)G^i_j &= BH^i_j, \quad i \in \{1,...,m\}, j \in \{1,..,i\}
\end{align*}
and all other coefficients in $\Phi_x$ and $\Phi_x^*$ are zero.
%
The above implies that
$G_l^* = -\sum_{i=l}^m J(\lambda-q)^{-(i+1-l)} BH_i^*$ for all
$l \in \{1,...,m\}$.
Define $G_{(l,i)}^* = -J(\lambda-q)^{-(i+1-l)}BH_i^*$
for all $l \in \{1,...,m\}$ and $i \in \{l,...,m\}$,
and note that
$G_l^* = \sum_{i=l}^m G_{(l,i)}^*$.
Write
$G^i_j = c^i_j(G_{(i,i)}^* + \Delta G^i_j)$.
Then, by \eqref{eq:phiu}
\begin{align*}
 J(\lambda-q)G_{(i,i)}^* &= -BH_i^*
  = -\f{1}{c^i_j}BH^i_j
  = \f{1}{c^i_j}J(\lambda-p^i_j)G^i_j \\
  &= J(\lambda-p^i_j)(G_{(i,i)}^* + \Delta G^i_j) \\
  &= (J(\lambda-q) - (p^i_j-q)I)(G_{(i,i)}^* + \Delta G^i_j)
\end{align*}
so $0 =  -(p^i_j-q)G_{(i,i)}^* + J(\lambda-p^i_j) \Delta G^i_j$
and \\
$\Delta G^i_j = (p^i_j-q)J(\lambda-p^i_j)^{-1}G_{(i,i)}^*$.
In summary,
\begin{align*}
  \Phi_x^* &= \sum\nolimits_{l=1}^m G_l^* \f{1}{(z-q)^l}
  = \sum\nolimits_{i=1}^m \sum\nolimits_{l=1}^i G_{(l,i)}^* \f{1}{(z-q)^l} \\
  G_{(l,i)}^* &= -J(\lambda-q)^{-(i+1-l)}BH_i^*, \quad l \in \{1,...,m\} \\
  \Phi_x &= \sum\nolimits_{i=1}^m \sum\nolimits_{j=1}^i G^i_j \f{1}{z-p^i_j} \\
  G^i_j &= c^i_j(I + (p^i_j-q)J(\lambda-p^i_j)^{-1})G_{(i,i)}^*.
\end{align*}
For any $i \in \{1,...,m\}$ and $l \in \{2,...,i\}$ write \\
  $J(\lambda-p^i_j)^{-1} G_{(l,i)}^* = G_{(l-1,i)}^* + \Delta G$.
Note that for $i \in \{1,...,m\}$ and $l \in \{2,...,i\}$,
  $J(\lambda-q)G_{(l-1,i)}^* =  G_{(l,i)}^*$.
Then
\begin{align*}
  &J(\lambda-q)G_{(l-1,i)}^* =  G_{(l,i)}^*
  = J(\lambda-p^i_j)(G_{(l-1,i)}^* + \Delta G) \\
  &= (J(\lambda-q) - (p^i_j-q)I)(G_{(l-1,i)}^* + \Delta G)
\end{align*}
so
$0 =  - (p^i_j-q)G_{(l-1,i)}^* + J(\lambda-p^i_j)\Delta G$
and \\
  $\Delta G = (p^i_j-q)J(\lambda-p^i_j)^{-1}G_{(l-1,i)}^*$
which implies that \\
  $J(\lambda-p^i_j)^{-1} G_{(l,i)}^* =
  (I + (p^i_j-q)J(\lambda-p^i_j)^{-1})G_{(l-1,i)}^*$.
Applying this equation recursively implies that
for $i \in \{1,...,m\}$ and $j \in \{1,...,i\}$
\begin{align*}
  &G^i_j = c^i_jG_{(i,i)}^* + c^i_j(p^i_j-q)J(\lambda-p^i_j)^{-1}G_{(i,i)}^* \\
  &= c^i_jG_{(i,i)}^*
  + c^i_j(p^i_j-q)(I + (p^i_j-q)J(\lambda-p^i_j)^{-1})G_{(i-1,i)}^* \\
  &= c^i_jG_{(i,i)}^* + c^i_j(p^i_j-q)G_{(i-1,i)}^* \\
  &+ c^i_j(p^i_j-q)^2J(\lambda-p^i_j)^{-1}G_{(i-1,i)}^*
  = ... \\
  &= \sum_{l=1}^i c^i_j (p^i_j-q)^{i-l} G_{(l,i)}^*
  + c^i_j (p^i_j-q)^i J(\lambda-p^i_j)^{-1}G_{(1,i)}^*.
\end{align*}
Therefore,
\begin{align*}
  \Phi_x(z) &= \sum_{i=1}^m \sum_{j=1}^i G^i_j \f{1}{z-p^i_j} \\
  &
  \overset{\substack{\text{above} \\ \text{identity}}}{=}    
  \sum_{i=1}^m \sum_{j=1}^i \sum_{l=1}^i c^i_j (p^i_j-q)^{i-l} G_{(l,i)}^*
  \f{1}{z-p^i_j} \\
  &+ \sum_{i=1}^m \sum_{j=1}^i c^i_j (p^i_j-q)^i J(\lambda-p^i_j)^{-1}G_{(1,i)}^*
  \f{1}{z-p^i_j} \\
  &
  \overset{\substack{\text{reverse} \\ \text{sum order}}}{=}
  \sum_{i=1}^m \sum_{l=1}^i G_{(l,i)}^*
  \left(\sum_{j=1}^i c^i_j (p^i_j-q)^{i-l} \f{1}{z-p^i_j} \right) \\
  &+ \sum_{i=1}^m 
  \left(\sum_{j=1}^i c^i_j (p^i_j-q)^i J(\lambda-p^i_j)^{-1} \f{1}{z-p^i_j}\right)
  G_{(1,i)}^* \\
  &
  \overset{\text{Lemma}~\ref{lem:easy}\text{(a)}}{=}    
  \sum_{i=1}^m \sum_{l=1}^i G_{(l,i)}^* \f{(z-q)^{i-l}}{\prod\limits_{j=1}^i (z-p^i_j)}\\
  &+ \sum_{i=1}^m 
  \left(\sum_{j=1}^i c^i_j (p^i_j-q)^i J(\lambda-p^i_j)^{-1} \f{1}{z-p^i_j}\right)
  G_{(1,i)}^*.
\end{align*}
Thus,
\begin{align*}
  \Phi_x(z) - \Phi_x^*(z)
  &= \sum_{i=1}^m \sum_{l=1}^i G_{(l,i)}^*
  \left(\f{(z-q)^{i-l}}{\prod\limits_{j=1}^i (z-p^i_j)} - \f{1}{(z-q)^l} \right) \\
  &\mkern-80mu + \sum_{i=1}^m 
  \left(\sum_{j=1}^i c^i_j (p^i_j-q)^i J(\lambda-p^i_j)^{-1} \f{1}{z-p^i_j}\right)
  G_{(1,i)}^*.
\end{align*}
This implies 
\begin{align*}
  &\left|\left|\Phi_x(z) - \Phi_x^*(z) \right|\right|_2 \\
  &
  \overset{\substack{\text{triangle} \\ \text{inequality}}}{\leq}  
  \sum_{i=1}^m \sum_{l=1}^i ||G_{(l,i)}^*||_2
  \left|\f{(z-q)^{i-l}}{\prod_{j=1}^i (z-p^i_j)} - \f{1}{(z-q)^l} \right| \\
  &+ \sum_{i=1}^m ||G_{(1,i)}^*||_2
  \left|\left|
  \sum_{j=1}^i c^i_j (p^i_j-q)^i J(\lambda-p^i_j)^{-1} \f{1}{z-p^i_j}\right|
  \right|_2 \\
  &
  \overset{\substack{\text{Lemma}~\ref{lem:abs} \\
      \text{Corollary}~\ref{cor:hard}\text{(b)}}}{\leq}  
  K D(\p) \\
  K &= \sum\nolimits_{i=1}^m \sum\nolimits_{l=1}^i ||G_{(l,i)}^*||_2 k_{(l,i)}
  + \sum\nolimits_{i=1}^m ||G_{(1,i)}^*||_2 k_i
\end{align*}
which proves \eqref{eq:new_goal} for Case 1.

It will be useful to derive an additional bound for use in the proof of Case 2.
In particular, we want to show that there exist constants $K_i > 0$
for $i \in \{1,...,m\}$ such that
\begin{align}
  \left|\left|\sum_{j=1}^i  G^i_j - G_{(1,i)}^* \right|\right|_2
  \leq K_i D(\p).
  \label{eq:goal_ki}
\end{align}
Write
  $G^i_j = c^i_j(G_{(1,i)}^* + \Delta G)$.
Then
\begin{align*}
  J(\lambda-q)^iG_{(1,i)}^* &= -BH_i^*
  = \f{1}{c^i_j} J(\lambda-p^i_j) G^i_j \\
  &= J(\lambda-p^i_j)(G_{(1,i)}^* + \Delta G)
\end{align*}
so
  $\Delta G = J(\lambda-p^i_j)^{-1}J(\lambda-q)^iG_{(1,i)}^* - G_{(1,i)}^*$,
which implies that
  $G^i_j = c^i_j J(\lambda-p^i_j)^{-1}J(\lambda-q)^iG_{(1,i)}^*$.
Thus,
\begin{align*}
  &\left|\left|\sum_{j=1}^i  G^i_j - G_{(1,i)}^* \right|\right|_2
  = \left|\left|\sum_{j=1}^i c^i_j J(\lambda-p^i_j)^{-1}J(\lambda-q)^iG_{(1,i)}^*
  \right.\right.
  \\ & \mkern+180mu \left.\left. \vphantom{\sum_{j=1}^i}
  - J(\lambda-q)^{-i} J(\lambda-q)^i G_{(1,i)}^* \right|\right|_2 \\
  &= \left|\left|\left(\sum_{j=1}^i c^i_j J(\lambda-p^i_j)^{-1} -
  J(\lambda-q)^{-i} \right) J(\lambda-q)^iG_{(1,i)}^* \right|\right|_2 \\
  &\leq \left|\left| \sum_{j=1}^i c^i_j J(\lambda-p^i_j)^{-1} - J(\lambda-q)^{-i}
  \right|\right|_2  
  \left|\left|J(\lambda-q)^iG_{(1,i)}^* \right|\right|_2 
\end{align*}
Therefore, in order to prove \eqref{eq:goal_ki} it suffices to show that
\begin{align}
  \left|\left| \sum_{j=1}^i c^i_j J(\lambda-p^i_j)^{-1} - J(\lambda-q)^{-i}
  \right|\right|_2 \leq K_i' D(\p)
  \label{eq:goal_kd}
\end{align}
for some constants $K_i' > 0$.
For $l \in \{0,...,m_q-1\}$, the $l$th superdiagonal of
$\sum_{j=1}^i c^i_j J(\lambda-p^i_j)^{-1}$ is given by
\begin{align*}
  &\sum_{j=1}^i c^i_j (-1)^l (\lambda-p^i_j)^{-(l+1)}
  = (-1)^l \sum_{j=1}^i c^i_j (\lambda-p^i_j)^{-(l+1)} \\
  & 
  =  {i-1 + l \choose l} \f{(-1)^l}{(\lambda-q)^{(i+l)}} + \epsilon_{(i,l)}, \quad
  |\epsilon_{(i,l)}| \leq K_{(i,l)} D(\p) 
\end{align*}
where we evaluate the sum by Corollary~\ref{cor:hard}(a).
Consider the function $f(x) = x^{-i}$ and note that
$f(J(\lambda-q)) = J(\lambda-q)^{-i}$.
By \cite[Theorem~11.1.1]{Go96},
for $l \in \{0,...,m_q-1\}$, the $l$th superdiagonal of
$J(\lambda-q)^{-i} = f(J(\lambda-q))$ is given by
\begin{align*}
  \f{1}{l!} f^{(l)}(\lambda-q)
  &= \f{i^{(l)}}{l!}\f{(-1)^l}{(\lambda-q)^{(i+l)}} 
  =  {i-1+l \choose l} \f{(-1)^l}{(\lambda-q)^{(i+l)}}.
\end{align*}
Thus, for each $i \in \{1,...,m\}$ and
$l \in \{0,...,m_q-1\}$, the difference between terms in superdiagonal $l$
of the matrix in \eqref{eq:goal_kd} is $\epsilon_{(i,l)}$, which satisfies
$|\epsilon_{(i,l)}| \leq K_{(i,l)} D(\p)$.
Therefore, by Fact 1 in the proof of Corollary~\ref{cor:hard},
this implies that
there exist $K_i' > 0$ such that
\eqref{eq:goal_kd} holds.

Case 2: $q = \lambda$.
Let $\hat{\mathcal{Q}}$ denote the poles in $\Phi_x$.
Substituting \eqref{eq:phix}, \eqref{eq:phiu},
and \eqref{eq:phiux} into \eqref{eq:bh}-\eqref{eq:zero} and
\eqref{eq:final_cons} implies that
\begin{align*}
  J(0)G_{m_q+m}^* &= 0, \quad G_1^* = V
  - \sum\nolimits_{\substack{\hat{q} \in \mathcal{Q} \\ \hat{q} \neq q}}
    G_{(\hat{q},1)}^* \\
    G_{i+1}^* &= J(0)G_i^*, \quad i \in \{m+1,...,m_q+m-1\} \\
    G_{i+1}^* &= J(0)G_i^* + BH_i^*, \quad i \in \{1,...,m\} \\
    J(0)G_{m_q+\tilde{m}} &= 0, \quad
    G_1 = V - \sum\nolimits_{\substack{\hat{q} \in \hat{\mathcal{Q}} \\ \hat{q} \neq q}}
    G_{(\hat{q},1)} \\
    G_{i+1} &= J(0)G_i, \quad i \in \{2,...,m_q+\tilde{m}-1\} \\
    G_2 &= J(0)G_1 + \tilde{m} \sum\nolimits_{i=1}^m BH^i_1 \\
    -J(q-p^i_j)G^i_j &= BH^i_j, \quad i \in \{1+\tilde{m},...,m\}, \\
                            & \mkern+85mu j \in \{1+\tilde{m},...,i\}
\end{align*}
where $G_{(\hat{q},1)}^*$ and $G_{(\hat{q},1)}$ denote the coefficients
of $\f{1}{z-\hat{q}}$ in $\Phi_x^*$ and $\Phi_x$, respectively, for the pole
$\hat{q}$.
For $l \in \{1,...,m_q\}$, define
Define
  $\hat{G}_1 = G_1 + \sum_{i=1+\tilde{m}}^m \sum_{j=1+\tilde{m}}^i G^i_j$, 
  $\hat{G}_l = J(0)^{l-1}\hat{G}_1$, and $\hat{G}_l^* = J(0)^{l-1}G_1^*$.
By \eqref{eq:goal_ki} from Case 1,
\begin{align*}
  ||\hat{G}_1-G_1^*||_2
  &\leq \sum\nolimits_{\substack{\hat{q} \in \mathcal{Q} \\ \hat{q} \neq q}}
  \sum\nolimits_{i=1}^{m_{\hat{q}}}
  \left|\left|G_{(\hat{q},i,1)}^*- \sum\nolimits_{j=1}^i G_{(\hat{q},i,1)}^j
  \right|\right|_2
  \\
  &\leq K_1 D(\p), \quad
  K_1 = \sum\nolimits_{\substack{\hat{q} \in \mathcal{Q} \\ \hat{q} \neq q}}
  \sum\nolimits_{i=1}^{m_{\hat{q}}} K_{(\hat{q},i)}.
\end{align*}
Thus, for $l \in \{1,...,m_q\}$ we have
\begin{align*}
  ||\hat{G}_l-\hat{G}_l^*||_2 
  &\leq ||J(0)^{l-1}||_2 ||\hat{G}_l-\hat{G}_l^*||_1
  \leq K_l D(\p)
\end{align*}
where $K_l = ||J(0)^{l-1}||_2K_1$.
For $l \in \{2,...,m_q+\tilde{m}\}$ define
\begin{align}
  \label{eq:gl}  
    \begin{split}
  \tilde{G}_1 &= G_1 - \hat{G}_1 =
  -\sum\nolimits_{i=1+\tilde{m}}^m \sum\nolimits_{j=1+\tilde{m}}^i  G^i_j \\
  \tilde{G}_l &= G_l - \hat{G}_l =
  -J(0)^{l-1} \sum\nolimits_{i=1+\tilde{m}}^m \sum\nolimits_{j=1+\tilde{m}}^i  G^i_j \\
  &+ J(0)^{l-2} \tilde{m} \sum\nolimits_{i=1}^m BH^i_1
    \end{split}
    \\
    \label{eq:gls}      
    \begin{split}
  \tilde{G}_1^* &= G_1^* - \hat{G}_1^* = 0, \;
  \tilde{G}_l^* = G_l^* - \hat{G}_l^* = \sum\nolimits_{i=1}^{\min\{l-1,m\}} G_{(l,i)}^*
  \\
  & G_{(l,i)}^* = J(0)^{l-(i+1)}BH_i^*, \quad i \in \{1,...,\min\{l-1,m\}\}.
    \end{split}
\end{align}
Then we have
\begin{align}
  \begin{split}
  &\left|\left|
  \Phi_x(z) - \Phi_x^*(z)
  \right|\right|_2 \\
  &\leq
  \left|\left|
  \sum_{i=1}^{m_q+\tilde{m}} \tilde{G}_i \f{1}{(z-q)^i}
  + \sum_{i=1+\tilde{m}}^m \sum_{j=1+\tilde{m}}^m G^i_j \f{1}{(z-p^i_j)}
  \right.\right. \\
  &\quad\;\;\left.\left.
  - \sum_{i=1}^{m_q+m} \tilde{G}_i^* \f{1}{(z-q)^i}
    \right|\right|_2
   +   \left|\left|
  \sum_{i=1}^{m_q} (\hat{G}_i-\hat{G}_i^*) \f{1}{(z-q)^i}
  \right|\right|_2 \\
  &\leq 
  \left|\left|
  \sum_{i=1}^{m_q+\tilde{m}} \tilde{G}_i \f{1}{(z-q)^i}
  + \sum_{i=1+\tilde{m}}^m \sum_{j=1+\tilde{m}}^m G^i_j \f{1}{(z-p^i_j)}
  \right.\right. \\
  &\left.\left.- \sum_{i=1}^{m_q+m} \tilde{G}_i^* \f{1}{(z-q)^i}
    \right|\right|_2
    + \hat{K} D(\p), \quad
    \hat{K} = \sum_{l=1}^{m_q} K_l \f{1}{(1-|q|)^l}.
  \end{split}
  \label{eq:suffice}
\end{align}
So, for the remainder of the proof let $\Phi_x$ and $\Phi_x^*$ denote
\begin{align}
  \Phi_x(z) &= \sum_{i=1}^{m_q+\tilde{m}} \tilde{G}_i \f{1}{(z-q)^i}
  + \sum_{i=1+\tilde{m}}^m \sum_{j=1+\tilde{m}}^m G^i_j \f{1}{(z-p^i_j)}
  \label{eq:phixt}
    \\
    \Phi_x^*(z) &= \sum_{i=1}^{m_q+m} \tilde{G}_i^* \f{1}{(z-q)^i},
    \label{eq:phixs}
\end{align}
and let $G_i$ and $G_i^*$ denote
$\tilde{G}_i$ and $\tilde{G}_i^*$, respectively.
Thus, to prove \eqref{eq:new_goal}, by \eqref{eq:suffice} it suffices to show
that there exists $K > 0$ such that
$\left|\left|\Phi_x(z)-\Phi_x^*(z)\right|\right|_2 \leq K D(\p)$.
By \eqref{eq:phiu} we have
\begin{align}
  G^i_j &= -J(q-p^i_j)^{-1}BH^i_j
  = -c^i_jJ(q-p^i_j)^{-1}BH_i^* \label{eq:gij_true}
\end{align}
for all $i \in \{\tilde{m}+1,...,m\}$ and $j \in \{\tilde{m}+1,...,i\}$.
We compute
\begin{align*}
  &\Phi_x(z) 
  \overset{\substack{\eqref{eq:gl} \\ \eqref{eq:phixt}}}{=}
  \sum_{i=1+\tilde{m}}^m \sum_{j=1+\tilde{m}}^i G^i_j \f{1}{z-p^i_j}
  \nonumber\\
  &+ \tilde{m} \sum_{l=2}^{m_q+1}  J(0)^{l-2} \sum_{i=1}^m BH^i_1 \f{1}{(z-q)^l}
  \nonumber \\
  &- \sum_{l=1}^{m_q} J(0)^{l-1} \sum_{i=1+\tilde{m}}^m \sum_{j=1+\tilde{m}}^i
  G^i_j \f{1}{(z-q)^l} \nonumber \\
  &
  \overset{\text{Lemma}~\ref{lem:case2}}{=}
  \left( \sum_{l=0}^{m_q-1} J(0)^l \tilde{m} \f{1}{(z-q)^{l+2}} \right) BH_1^*
  \\
  &+ \sum_{i=1+\tilde{m}}^m \left(
  \sum_{j=1+\tilde{m}}^i
  J(q-p^i_j)^{-1} \f{-c^i_j}{z-p^i_j} \f{(p^i_j-q)^{m_q}}{(z-q)^{m_q}} \right.
  \\
  &+ \sum_{l=0}^{m_q-i-1}  \sum_{k=0}^{m_q-i-1-l} J(0)^{l}
  \sum_{j=1+\tilde{m}}^i \f{c^i_j  (p^i_j-q)^{m_q-2-k-l}}{(z-q)^{m_q-k}} \\
  &\left. \vphantom{\sum_{l=1}^2}
  + \tilde{m} c^i_1 J(0)^{m_q-1} \f{1}{(z-q)^{m_q+1}} \right) BH_i^*.
\end{align*}
By \eqref{eq:phixs} and \eqref{eq:gls} we have
\begin{align*}
  \Phi_x^* = \sum_{i=1}^m \sum_{l=0}^{m_q-1} J(0)^l BH_i^* \f{1}{(z-q)^{i+l+1}}.
\end{align*}
Therefore, for $i \in \{1,...,m\}$ and $l \in \{0,...,m_q-1\}$,
the $l$th superdiagonal
of the term multiplying $BH_i^*$ in $\Phi_x^*$ is given by
$\f{1}{(z-q)^{i+l+1}}$.
Thus, by Lemma~\ref{lem:case2}, for every $j, j' \in \{1,...,m_q\}$,
\begin{align*}
  &\left|(\Phi_x(z)-\Phi_x^*(z))_{(j,j')} \right| \leq
  \sum_{i=1}^m \sum_{l=0}^{m_q-j} K_{i,l} D(\p) |(BH_i^*)_{(j+l,j')}| \\
 &\leq K^{(j,j')} D(\p), \quad
  K^{(j,j')} = \sum_{i=1}^m \sum_{l=0}^{m_q-1} K_{i,l} ||BH_i^*||_F.
\end{align*}
By Fact 1 in the proof of Corollary~\ref{cor:hard},
this implies that $||\Phi_x(z)-\Phi_x^*(z)||_2 \leq K D(\p)$
for some $K > 0$, which proves \eqref{eq:new_goal} for Case 2.
\end{proof}

Theorem~\ref{thm:gen} applies the approximation error bounds of
Lemma~\ref{lem:bound} to the optimal
solution of \eqref{eq:spa_inf} to obtain the desired suboptimality bounds.

\begin{proof}[Proof of Theorem~\ref{thm:gen}]
Let $(\Phi_x^*,\Phi_u^*)$ be an optimal solution to \eqref{eq:spa_inf}.
By Lemma~\ref{lem:bound}, there exist $\Phi_x, \Phi_u \in \f{1}{z} \h$
which are a feasible solution to \eqref{eq:spa_obj}-\eqref{eq:spa_imp} and
satisfy the approximation error bounds \eqref{eq:ubound}-\eqref{eq:xbound}.
Letting $J(\Phi_x,\Phi_u)$ denote the value of the objective of
\eqref{eq:spa_inf} for $(\Phi_x, \Phi_u)$,
we compute
\begin{align*}
  J(\p) & 
  \overset{\substack{\text{definition} \\ \text{of optimum}}}{\leq}
  J(\Phi_x,\Phi_u) \\
  &
  \overset{\substack{\text{adding} \\ \text{zero}}}{=}
    \left|\left|
    C(\Phi_x(z)-\Phi_x^*(z) + \Phi_x^*(z))\hat{B} \right.\right.
    \\
    &\left.\left.
    + D(\Phi_u(z)-\Phi_u^*(z) + \Phi_u^*(z))\hat{B}
    - T_{\text{desired}}(z)
    \right|\right|_{\mathcal{H}_2} \\
    & + \lambda \left|\left|
    C(\Phi_x(z)-\Phi_x^*(z) + \Phi_x^*(z))\hat{B} \right.\right.\\
    &\left.\left. + D(\Phi_u(z)-\Phi_u^*(z)+\Phi_u^*(z))\hat{B}
    - T_{\text{desired}}(z)   
    \right|\right|_{\mathcal{H}_\infty} \\
    & 
    \overset{\substack{\text{triangle} \\ \text{inequality}}}{\leq}
    \left|\left|C\Phi_x^*(z)\hat{B} + D\Phi_u^*(z)\hat{B}
    - T_{\text{desired}}(z) \right|\right|_{\mathcal{H}_2}  \\
    & + \lambda \left|\left|
    C\Phi_x^*(z)\hat{B} + D\Phi_u^*(z)\hat{B} - T_{\text{desired}}(z)
    \right|\right|_{\mathcal{H}_\infty} \\
    &+ \left|\left|
    C(\Phi_x(z)-\Phi_x^*(z))\hat{B}\right|\right|_{\mathcal{H}_2} \\
    &+ \left|\left|D(\Phi_u(z)-\Phi_u^*(z))\hat{B}
    \right|\right|_{\mathcal{H}_2} \\
    & + \lambda \left|\left|
    C(\Phi_x(z)-\Phi_x^*(z))\hat{B}
    \right|\right|_{\mathcal{H}_\infty} \\
    &+ \lambda \left|\left|D(\Phi_u(z)-\Phi_u^*(z))\hat{B}
    \right|\right|_{\mathcal{H}_\infty} \\
    & 
    \overset{\substack{\eqref{eq:ubound} \\ \eqref{eq:xbound}}}{\leq}
    J^* + K D(\p) \\
    K &= ||C||_F K^x_2 ||\hat{B}||_F
    + ||D||_F K^u_2 ||\hat{B}||_F \\
    &+ \lambda ||C||_2 K^x_\infty ||\hat{B}||_2
    + \lambda ||D||_2 K^u_\infty ||\hat{B}||_2.
\end{align*}
This implies that
\begin{align*}
  \f{J(\p) - J^*}{J^*} \leq \f{K}{J^*} D(\p)
\end{align*}
which is the desired bound, and where
$K = K(\mathcal{Q},G_{(q,j)}^*,H_{(q,j)}^*,r,\delta)$ by the proofs of
\cite[Theorem 1]{Fi22a} and Lemma~\ref{lem:bound}.
\end{proof}

\begin{proof}[Proof of Corollary~\ref{cor:spiral}]
  Combining Theorem~\ref{thm:gen} with the result and proof of
  \cite[Theorem~5]{Fi22a} yields the desired result.
\end{proof}

\section{Conclusion}\label{sec:con}

This work combined SLS with SPA to develop a new control design method.
Unlike DBC, SPA does not result in deadbeat control, feasibility
is automatic so it does not require slack variables
which lead to additional suboptimality, and it can be solved
by a single SDP, as opposed to the iterative algorithm that DBC requires.
A suboptimality certificate was provided for SPA 
which, unlike the DBC bound, does not require a sufficiently long time
horizon that the optimal impulse response has already decayed, and does not
depend on this decay rate.
The bound is specialized for the Archimedes spiral pole selection \cite{Fi22a}.
An example 
shows that SPA achieves much
better matching with the optimal solution than DBC with orders of magnitude
fewer poles.
Future work should address extensions to state and
input constraints, application of SPA to output feedback,
and extensions to continuous-time.

\bibliographystyle{ieeetr}
\bibliography{refs}

\vspace{-20pt}


\begin{IEEEbiography}
  [{\includegraphics[width=1in,height=1.25in,clip,keepaspectratio]
      {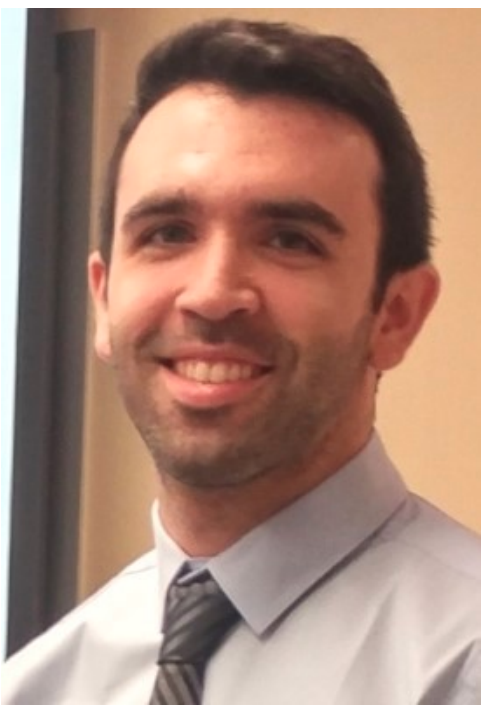}}]{Michael W. Fisher} is an Assistant Professor
  in the Department of Electrical and Computer Engineering at the University of
  Waterloo, Canada.  He was a postdoctoral researcher with
  the Automatic Control and Power System Laboratories at
  ETH Zurich.  He received his Ph.D. in Electrical Engineering:
  Systems at the University of Michigan, Ann Arbor in 2020, and a
  M.Sc. in Mathematics from the same institution in 2017. He received
  his B.A. in Mathematics and Physics from Swarthmore College in 2014.
  His research interests are in dynamics, control, and optimization of
  complex systems, with an emphasis on electric power systems.
  He was a finalist for the 2017 Conference on Decision and Control (CDC)
  Best Student Paper Award and a recipient
  of the 2019 CDC Outstanding Student Paper Award.
\end{IEEEbiography}

\vspace{-20pt}


\begin{IEEEbiography}
  [{\includegraphics[width=1in,clip,keepaspectratio]
      {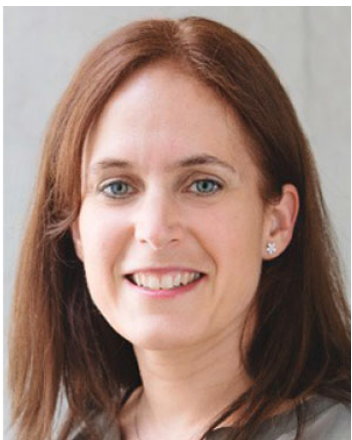}}]{Gabriela Hug} was born in Baden,
  Switzerland. She received the M.Sc. degree in electrical engineering
  and the Ph.D. degree from the Swiss Federal Institute of Technology
  (ETH), Zurich, Switzerland, in 2004 and 2008, respectively. After
  the Ph.D. degree, she worked with the Special Studies Group of Hydro
  One, Toronto, ON, Canada, and from 2009 to 2015, she was an
  Assistant Professor with Carnegie Mellon University, Pittsburgh, PA,
  USA. She is currently a Professor with the Power Systems Laboratory,
  ETH Zurich. Her research is dedicated to control and optimization of
  electric power systems.
\end{IEEEbiography}

\vspace{-20pt}


\begin{IEEEbiography}
  [{\includegraphics[width=1in,height=1.25in,clip,keepaspectratio]
      {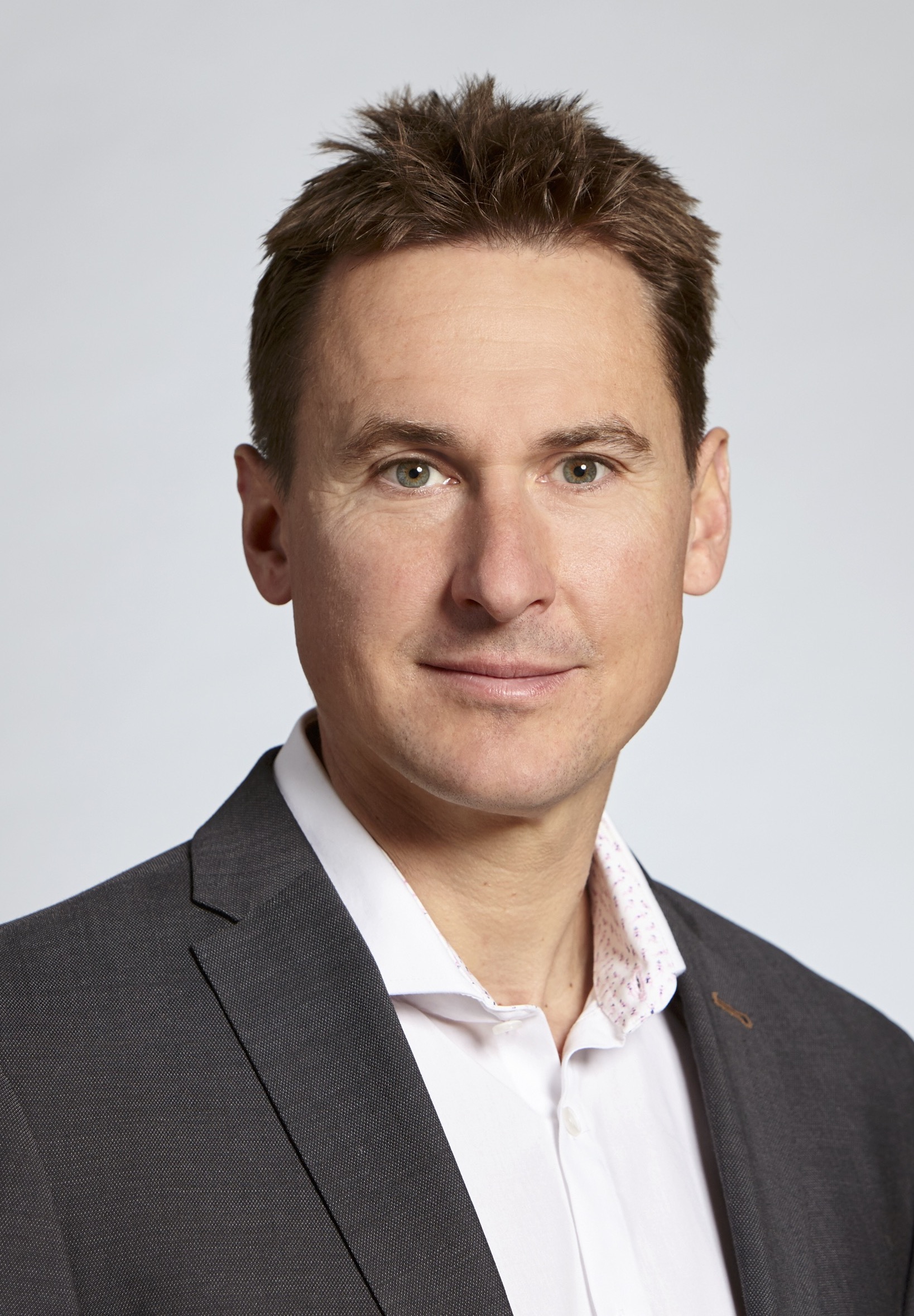}}]{Florian D\"{o}rfler} is an Associate
  Professor at the Automatic Control Laboratory at ETH Zurich,
  Switzerland, and the Associate Head of the Department of Information
  Technology and Electrical Engineering. He received his Ph.D. degree
  in Mechanical Engineering from the University of California at Santa
  Barbara in 2013, and a Diplom degree in Engineering Cybernetics from
  the University of Stuttgart, Germany, in 2008. From 2013 to 2014 he
  was an Assistant Professor at the University of California
  Los
  Angeles.
  His primary research interests are centered around control,
  optimization, and system theory with applications in network
  systems, especially electric power grids.
  He is a recipient of the
  distinguished young research awards by IFAC (Manfred Thoma Medal
  2020) and EUCA (European Control Award 2020).
  He and his students were winners or finalists of numerous Best Paper
  Awards.
  His students were
  winners or finalists for Best Student Paper awards at the European
  Control Conference (2013, 2019), the American Control Conference
  (2016), the Conference on Decision and Control (2020), the PES
  General Meeting (2020), the PES PowerTech Conference (2017), and the
  International Conference on Intelligent Transportation Systems
  (2021).
  He is furthermore a recipient of the 2010 ACC Student Best
  Paper Award, the 2011 O. Hugo Schuck Best Paper Award, the 2012-2014
  Automatica Best Paper Award, the 2016 IEEE Circuits and Systems
  Guillemin-Cauer Best Paper Award, and the 2015 UCSB ME Best PhD
  award.
\end{IEEEbiography}

\end{document}